\documentclass{article}

\PassOptionsToPackage{numbers, compress}{natbib}
\usepackage[nonanonymous, final, main]{neurips_2026}

\usepackage[utf8]{inputenc} 
\usepackage[T1]{fontenc}    
\usepackage{hyperref}       
\usepackage{url}            
\usepackage{booktabs}       
\usepackage{amsfonts}       
\usepackage{nicefrac}       
\usepackage{microtype}      
\usepackage{xcolor}         
\usepackage{amsmath}
\usepackage{graphicx}
\usepackage{wrapfig}
\usepackage{amsthm}
\newtheorem{lemma}{Lemma}[section]
\newtheorem{definition}{Definition}[section]
\newtheorem{theorem}{Theorem}[section]

\usepackage{silence}
\usepackage{listings}
\newcommand\code[1]{\lstinline$#1$}

\usepackage{algorithm}
\usepackage{algpseudocode}
\usepackage{multirow}
\usepackage{pifont}
\usepackage[normalem]{ulem}

\definecolor{commentcolor}{HTML}{a299a2}
\definecolor{string-color}{rgb}{0.3333, 0.5254, 0.345}
\definecolor{indentcolor}{rgb}{0.75, 0.75, 0.75}
 
\lstdefinestyle{pseudocode}{
    language=Python,
    mathescape=true,
    keywords={PACEStep,if,return,else,while,for,in,append,ceil,break,max,min,set,or,is,None},
    keywordstyle=\color{blue}\bfseries,
    morekeywords={[2], requests, T, M, utility, rho},
    keywordstyle={[2]\color{string-color}\bfseries},
    morekeywords={[3], def},
    keywordstyle={[3]\color{red}\bfseries},
    basicstyle=\small\ttfamily,
    identifierstyle=\color{black},
    sensitive=false,
    commentstyle=\color{commentcolor},
    columns=fullflexible,
    keepspaces=true,               
    breaklines=true,
    showstringspaces=false,
    numbers=left,
    numberstyle=\tiny,
    numbersep=4pt,
    xleftmargin=5pt,
    lineskip={-1pt},
    postbreak=\mbox{\textcolor{black}{$\hookrightarrow$}\space},
    stringstyle=\color{string-color},
    literate={|}{{{\color{indentcolor}\vrule width 0.4pt\hspace{4pt}}}}{1},
}

\makeatletter
\newenvironment{codealgorithm}[1][htb]{%
  \renewcommand{\ALG@name}{Algorithm}%
  \begin{algorithm}[#1]%
  }{\end{algorithm}}
\makeatother

\algnewcommand{\LineComment}[1]{\State \(\Comment\) #1}

\title{Pacing Branch Parallelism in LLM Serving}

\author{
  Swapnil Gandhi\textsuperscript{1}
  \quad
  Siva Kumar Sastry Hari\textsuperscript{2}
  \quad
  William J. Dally\textsuperscript{1,2}
  \quad
  Christos Kozyrakis\textsuperscript{1,2}\\[1ex]
  \textsuperscript{1}Stanford University \quad
  \textsuperscript{2}NVIDIA\\[1ex]
}

\begin{document}
\maketitle

\begin{abstract}
Recent methods expose intra-request parallelism in LLM outputs, allowing independent branches to decode concurrently. Existing serving systems execute these branches eagerly or under fixed caps. We show that both are brittle: eager admission inflates the shared decode step, degrading co-batched requests in serial stages, while conservative fixed caps forgo the throughput that motivated exposing branches in the first place. We call the excess step latency caused by admitted branches the \emph{branch externality} and show that the safe width depends on batch composition, context lengths, and accumulated slack, all of which change continuously over a workload trace. We introduce \textsc{Pace}, a per-step admission controller that treats extra branches as opportunistic width, admitted only when the predicted branch externality fits within the batch's current slack budget. Per-step regulation is practical because branch-level scheduling decouples compute from memory: branches share the request's prefix KV, so expanding or contracting width requires no memory reclamation. On Qwen3-32B, \textsc{Pace} improves goodput by $1.77\times$ over serving without branch parallelism and by $1.48\times$ over eager admission while maintaining 99\% SLO attainment. Comparable gains ($1.72$--$2.15\times$) hold on five additional dense and mixture-of-experts models with up to 671B parameters. Code is available at \url{https://github.com/Swapnil-Gandhi/PACE}.
\end{abstract}
\section{Introduction}\label{sec:introduction}

As LLM responses grow longer, driven by chain-of-thought reasoning~\cite{chain-of-thought, deepseek-r1}, structured generation, and multi-step tasks~\cite{react}, the sequential nature of autoregressive decoding becomes an increasingly visible bottleneck. A growing body of work addresses this by exposing \emph{intra-request parallelism} (IRP): decomposing a single response into independent branches that decode concurrently before a reduce phase combines them~\cite{skeleton-of-thought, pasta, apar, apr, aspd, multiverse, threadweaver}. For an isolated request, the appeal is immediate: parallel execution shrinks the critical path from the sum of branch lengths toward the length of the slowest branch, converting otherwise idle decode capacity into useful progress~\cite{parallel-prompt, paralleltextgeneration-survey}.

The difficulty begins when this private speedup is realized inside a shared serving system. Under continuous batching~\cite{orca}, the engine executes one forward pass per decode step across all active sequences, advancing each sequence by one token. Step latency grows with the number of sequences and their aggregate context length~\cite{servegen, nanoflow}. A request in a parallel stage contributes multiple branch sequences to the same step, widening the batch for everyone. That request is compensated: it receives progress on several branches. A co-batched request in a serial stage receives only its single next token, yet waits for the same inflated step. We call this excess latency the \emph{branch externality}. In \S\ref{sec:characterization}, we show that under moderate and high load, eager branch admission can raise raw throughput while simultaneously degrading goodput and SLO attainment.

Simple remedies do not resolve this tension. Disabling IRP eliminates the interference but forfeits the throughput and latency gains that motivated exposing branches in the first place. Fixed caps, which bound how many branches each request may decode per step, help in one operating regime and hurt in another, because the safe width depends on batch composition, context lengths, and accumulated slack, all of which change continuously over a workload trace. Reactive controllers that adjust a global cap from observed SLO violations fare better, but they respond only after violations accumulate and apply one width to every request, regardless of how close each is to its deadline. What is needed is a per-step admission controller that predicts the cost of a wider step before committing to it, widening when the current batch has headroom and contracting when requests in serial stages are at risk. It must also be cheap enough to revise its decision at every step.

Historically, such fine-grained regulation has been impractical because the natural scheduling unit is the request~\cite{vtc, llumnix}. Revising a request-level decision entangles compute and memory: pausing a request forces the scheduler either to keep its KV state resident, consuming scarce memory, or to evict it, incurring restoration cost when the request resumes~\cite{vllm, sglang}. Frequent per-step revisions are therefore expensive to undo.

\textbf{Our central observation is that branch-level scheduling separates these concerns.} During a parallel stage, all branches share the request's prefix KV~\cite{multiverse}. A branch's own KV consists only of the tokens it has generated so far, typically a small fraction of the shared prefix~\cite{apar}. When the scheduler defers a branch from a given step, it does not reclaim the prefix (which remains resident for the sibling branches that are still admitted) and it does not evict the branch-local KV (which is small enough to leave in place). Admitting the branch in a subsequent step requires no state restoration: the scheduler simply adds it back to the next forward pass. The residual cost is that a deferred branch's branch-local KV remains resident while idle, but this is bounded by the tokens the branch has generated, and the memory is accounted to the parent request. If KV cache pressure requires eviction, the serving system preempts the entire request through its normal eviction policy. Branch regulation is therefore cheap, voluntary, and reversible. This decoupling is what makes per-step width regulation practical.

We introduce \textsc{Pace} (\underline{P}er-step \underline{A}dmission \underline{C}ontrol for \underline{E}xternalities), a runtime controller that exploits this decoupling. \textsc{Pace} first protects \emph{single-token progress}: every active request is guaranteed one token in the next step. It then treats additional branches as \emph{opportunistic width}, admitting them only when the predicted branch externality fits within the batch's current slack budget. When slack is abundant, \textsc{Pace} expands toward eager execution; as requests in serial stages approach their latency targets, it contracts toward single-token progress. Because contraction amounts to declining branch admissions rather than evicting requests, \textsc{Pace} adapts every step without policy switches, preemption machinery, or fallback logic.

We make the following contributions:

\begin{itemize}
    \item We characterize the \emph{branch externality}: requests in parallel stages benefit from widened decode steps, while co-batched requests in serial stages pay the inflated latency without receiving additional progress, making eager, fixed-cap, and reactive policies brittle across operating points~(\S\ref{sec:characterization}).
    \item We show that branch-level scheduling \emph{decouples compute regulation from memory regulation}: branches share the prefix KV and have small branch-local state, so step width can be expanded or contracted per step without eviction or restoration. This is the structural property that makes per-step control feasible~(\S\ref{sec:PACE:phases}, \S\ref{sec:PACE:feedback}).
    \item We introduce \textsc{Pace}, a per-step controller that predicts the branch externality, bounds it against the batch's slack budget, and allocates opportunistic width via a greedy planner with a pluggable utility interface~(\S\ref{sec:PACE}).
    \item We evaluate \textsc{Pace} under mixed workloads on six models, two dense and four mixture-of-experts. On Qwen3-32B, it improves goodput by $1.77\times$ over disabling IRP and by $1.48\times$ over eager admission while maintaining 99\% SLO attainment, versus below 50\% for eager admission and 79\% for a reactive width controller~(\S\ref{sec:evaluation}).
\end{itemize}
\section{Characterizing Intra-Request Parallelism}
\label{sec:characterization}

\begin{figure}[t]
    \centering
    \includegraphics[width=\linewidth]{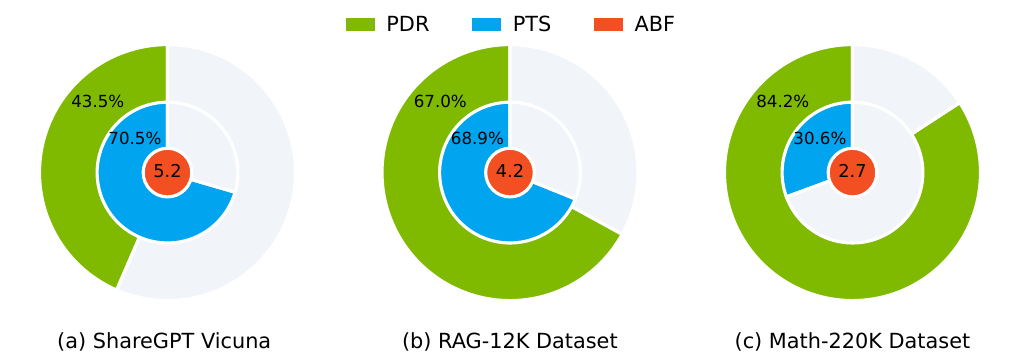}
    \vspace{-0.12in}
    \caption{\textbf{Intra-request parallelism across workloads.} Proportion of decomposable requests (PDR), parallel token share (PTS), and average branch fanout (ABF) for three datasets.}
    \label{fig:parallelism-in-the-wild}
    \vspace{-0.18in}
\end{figure}

Intra-request parallelism can substantially accelerate individual requests. This section shows that it also creates a cost that falls on other requests: under realistic mixed workloads, the acceleration of requests in parallel stages comes at the expense of co-batched requests in serial stages, and neither fixed nor reactive policies navigate the trade-off across operating points.

\subsection{Intra-request parallelism in the wild}
\label{sec:char:blp}

Several recent methods shorten the critical path of LLM decoding by exposing independent branches within a single response. Skeleton-of-Thought (SoT)~\citep{skeleton-of-thought} expands outline points concurrently; APAR~\citep{apar} emits explicit branch tokens; PASTA~\citep{pasta} introduces asynchronous promises; ASPD~\citep{aspd} trains branch-invisible attention masks for native parallel decoding; and Multiverse~\citep{multiverse} exposes Map/Process/Reduce control flow. These methods differ in how branches are discovered and represented, but they produce the same serving-visible structure~\citep{aspd}.

Output generation comprises a sequence of interleaved stages, each decoding in serial or parallel mode. In serial stages, generation proceeds through a single autoregressive continuation; in parallel stages, the model decodes multiple independent branches concurrently, each still autoregressive. A reduce phase combines the completed branches before generation continues, possibly entering another parallel stage later in the response.

We call a request \emph{decomposable} if it enters at least one parallel stage during its lifetime, and \emph{non-decomposable} otherwise. At any given decode step, a decomposable request may be in either a serial or a parallel stage. A request in a serial stage, whether decomposable or not, contributes a single continuation token and is indistinguishable from a non-decomposable request. The number of branches exposed during a parallel stage is the request's \emph{branch fanout} $n_r$.

\textbf{Prevalence.} \emph{How common are decomposable requests?} We characterize three representative datasets (Figure~\ref{fig:parallelism-in-the-wild}). The \emph{proportion of decomposable requests} (PDR) measures how often parallelism appears: Math-220K~\cite{math-220k} has the highest PDR at 84.2\%, followed by RAG-12K~\cite{rag-12k} at 67.0\% and ShareGPT Vicuna~\cite{vicuna} at 43.5\%. The \emph{parallel token share} (PTS) measures how much of a decomposable response is parallel: ShareGPT (70.5\%) and RAG-12K (68.9\%) are dominated by parallel tokens when they decompose, while Math-220K (30.6\%) produces short, narrow parallel stages. The \emph{average branch fanout} (ABF) indicates the typical maximum step width the runtime could admit: ShareGPT averages 5.2 branches per parallel stage, RAG-12K 4.2, and Math-220K 2.7.

The pattern is not uniform. Math reasoning decomposes frequently but narrowly: many parallel stages, each with few short branches covering a small share of the output. Structured tasks like ShareGPT and RAG-12K decompose less often but more deeply, with wider branching and a larger share of tokens in parallel stages. What determines the severity of the branch externality (\S\ref{sec:char:externality}) is not PDR alone but its combination with ABF and PTS: a high-PDR, low-ABF workload creates many small externalities; a moderate-PDR, high-ABF workload creates fewer but larger ones.

\textbf{Implications for serving.} Decomposable requests are common enough to matter at scale across all three workload types; the branch externality already erodes SLO attainment when only one request in ten is decomposable (Appendix~\ref{app:pdr-sensitivity}). Yet even in workloads with high PDR, requests spend a significant fraction of their lifetime in serial stages (PTS is always below 71\%). The typical batch will therefore contain a mix of requests in parallel and serial stages. It is in this mixed-batch setting that the cost of branch parallelism becomes visible.

\begin{figure*}[t]
    \centering
    \includegraphics[width=\linewidth]{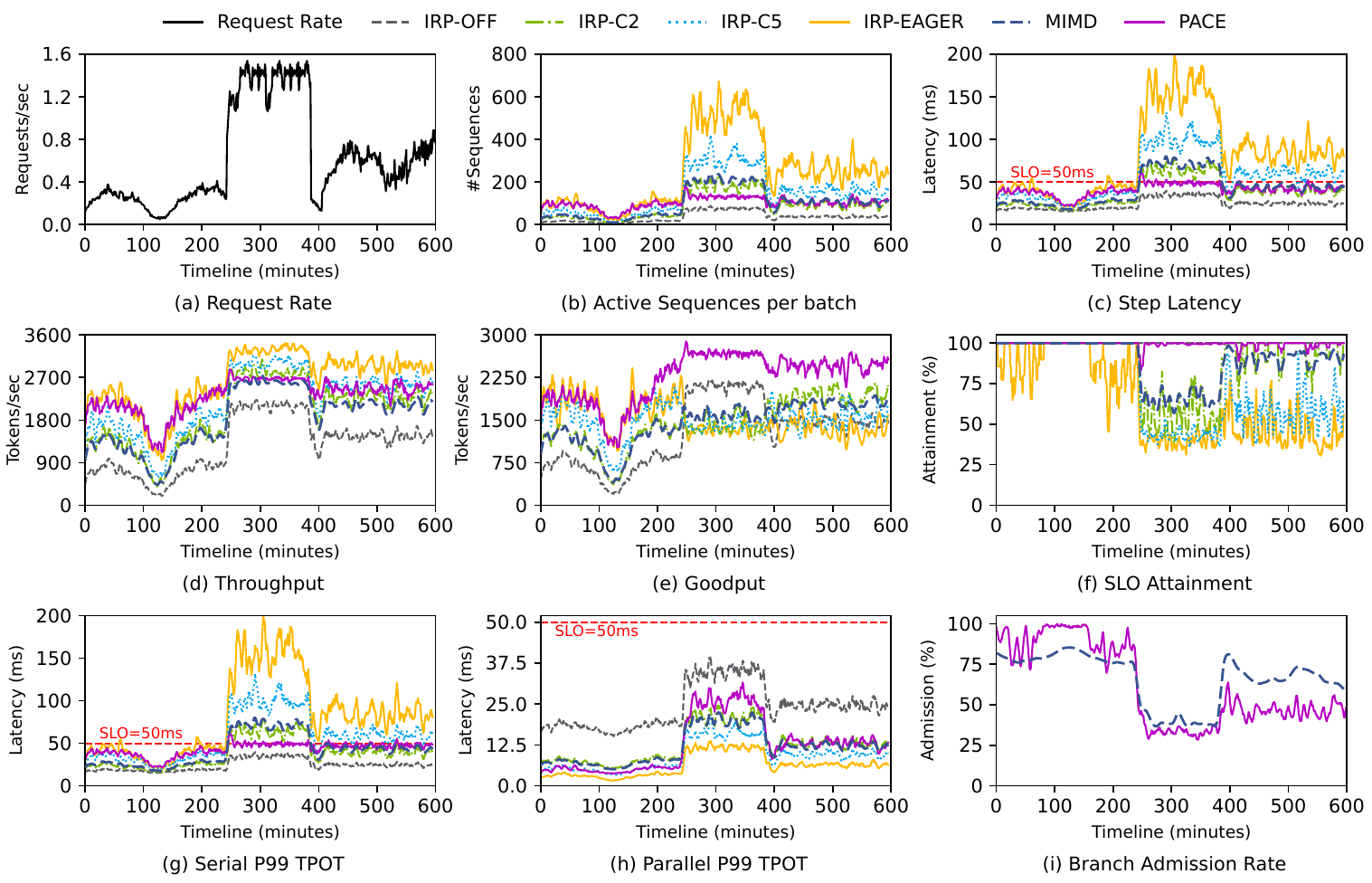}
    \vspace{-0.15in}
    \caption{\textbf{The throughput trap and its resolution.} Four fixed step-width policies, a reactive width controller (MIMD), and \textsc{Pace} on a mixed workload. \textsc{IRP-Eager} raises throughput but collapses goodput and SLO attainment under load, and the cost falls asymmetrically on requests in serial stages. \textsc{Pace} adjusts its branch admission rate every step (panel~(i)), retaining most of eager's throughput while protecting SLO attainment.}
    \label{fig:main-results}
    \vspace{-0.2in}
\end{figure*}

\subsection{The throughput trap}
\label{sec:char:trap}

The metrics above characterize parallelism in isolation. In a serving system, branches share a decode step with other requests, and the cost of admitting them is borne collectively. When the runtime increases a request's \emph{step width} $w_{r,t}$, the number of its branches included in step $t$, the step takes longer and every co-batched request pays that latency. The request in a parallel stage is compensated: it receives $w_{r,t}$ tokens of branch progress. A co-batched request in a serial stage receives only its single continuation token.

Figure~\ref{fig:main-results} makes this concrete on an Azure-derived trace~\citep{azure-llm-trace} (details in Appendix~\ref{app:experimental-setup}). A mixed workload runs under four fixed step-width policies: \textsc{IRP-Off} ($w_{r,t}=1$), \textsc{IRP-C2} and \textsc{IRP-C5} (which cap $w_{r,t}$ at 2 and 5), and \textsc{IRP-Eager} ($w_{r,t}=n_r$). A fifth policy, MIMD, adapts a global width cap reactively from observed SLO violations. At low load, larger step widths extract more raw throughput: the batch has headroom to absorb wider steps. As load grows, wider policies inflate step latency~(\hyperref[fig:main-results]{\ref*{fig:main-results}(c)}). Raw throughput continues to favor \textsc{IRP-Eager}~(\hyperref[fig:main-results]{\ref*{fig:main-results}(d)}), but goodput and SLO attainment collapse~(\hyperref[fig:main-results]{\ref*{fig:main-results}(e,\,f)}): the system produces more tokens while delivering fewer of them within their latency budget. We call this the \emph{throughput trap}: the metric operators inspect first (raw throughput) favors the policy that fails the metric users experience (goodput).

The per-class view reveals where the SLO failures land. Serial P99 TPOT~(\hyperref[fig:main-results]{\ref*{fig:main-results}(g)}) spikes sharply under \textsc{IRP-Eager} during stress events, crossing the 50\,ms target. Under \textsc{IRP-Off}, it stays flat. Parallel P99 TPOT~(\hyperref[fig:main-results]{\ref*{fig:main-results}(h)}) shows the opposite: it remains low regardless of load. The throughput trap is not a uniform degradation but an asymmetric transfer: requests in serial stages absorb the latency, while requests in parallel stages see their lowest P99 TPOT under \textsc{IRP-Eager}.

The intermediate caps do not escape. \textsc{IRP-C2} and \textsc{IRP-C5} are safer than \textsc{IRP-Eager} under load but still inflate step latency during stress events, and their goodput still collapses, just later. They also leave throughput unused at low load. The safe step width depends on conditions that change continuously over the trace: batch composition, context lengths, accumulated slack, and KV cache pressure. No fixed cap is right across operating points.

Reactive adaptation narrows the gap but does not close it. MIMD grows its width multiplicatively while steps meet the target and cuts it when violations accumulate. It reaches the best attainment of any branch-admitting baseline (79\%; Appendix~\ref{app:mimd}), yet it errs in both directions~(\hyperref[fig:main-results]{\ref*{fig:main-results}(i)}): at low load it admits fewer branches than the batch can absorb, and under moderate load it admits more than the most urgent requests can tolerate. Both errors are structural. A windowed controller acts only after violations have occurred, and a single global width cannot distinguish a request near its deadline from one with ample slack.

\subsection{Branch externality}
\label{sec:char:externality}

The throughput trap and its asymmetric incidence reduce to a single quantity. The asymmetry is a classic externality in the economic sense: one agent's resource consumption imposes uncompensated costs on others~\cite{pigou2017economics}. We formalize this as the \emph{branch externality}. Let $S_0$ be the baseline step in which every active request advances by one token. Let $S(\mathbf{a})$ be a widened step that admits additional branches according to an allocation $\mathbf{a} = (a_r)_{r \in \mathcal{R}}$ over the active requests $\mathcal{R}$, where $a_r \geq 0$ is the number of additional branches granted to request $r$, which then decodes with step width $w_{r,t} = 1 + a_r$ (\textsc{IRP-Off} sets $a_r = 0$; \textsc{IRP-Eager} sets $a_r = n_r - 1$). The branch externality is
\[
    E_t(\mathbf{a}) \;=\; T\!\bigl(S(\mathbf{a})\bigr) \;-\; T(S_0),
\]
where $T(\cdot)$ is the predicted step latency.

Every request in the batch pays $E_t(\mathbf{a})$; only requests in parallel stages with $a_r > 0$ receive additional progress in exchange. The step latency panel~(\hyperref[fig:main-results]{\ref*{fig:main-results}(c)}) makes the externality visible: the gap between each policy's line and \textsc{IRP-Off}'s baseline is that policy's $E_t$. \textsc{IRP-Eager}'s gap grows unchecked during stress events, reaching over 150\,ms above baseline; the fixed caps produce smaller but still unregulated gaps. The externality is monotone non-decreasing in each $a_r$: admitting more branches cannot shorten a decode step.

What the runtime needs is neither a tighter cap nor a faster reaction, but a per-step controller that predicts $E_t(\cdot)$ against the current batch state before committing to a step, and admits branches only when the predicted externality fits within the batch's available slack. The next section introduces \textsc{Pace}, such a controller. Its behavior is already visible in Figure~\ref{fig:main-results}: \textsc{Pace}'s step latency tracks near \textsc{IRP-Eager} at low load and stays near the SLO target during stress events~(\hyperref[fig:main-results]{\ref*{fig:main-results}(c)}), while its branch admission rate~(\hyperref[fig:main-results]{\ref*{fig:main-results}(i)}) shows the control signal adapting per step.
\section{\textsc{Pace}: Per-Step Admission Control for Externalities}
\label{sec:PACE}

This section builds the controller that \S\ref{sec:characterization} calls for. We formalize the parallel phase and show that per-step width decisions do not change model output~(\S\ref{sec:PACE:phases}), then develop a cost model~(\S\ref{sec:PACE:cost}), a slack-budgeted admission rule~(\S\ref{sec:PACE:slack}), and a greedy planner that combines them~(\S\ref{sec:PACE:planner}). We close by showing why per-step regulation is structurally cheap~(\S\ref{sec:PACE:feedback}).

\subsection{Parallel phases and schedule invariance}
\label{sec:PACE:phases}

A \emph{parallel phase}, the formal counterpart of a parallel stage (\S\ref{sec:char:blp}), is a request-local segment in which the frontend has exposed a finite set of independent branches that must all complete before a reduce phase combines them. We represent a phase as $\phi = (P_\phi, H_\phi, \{h_i\}_{i=1}^{n_\phi})$, where $P_\phi$ is the shared prefix generated before the phase, $H_\phi$ is the phase-level header (e.g., an outline or task decomposition), $h_i$ is the header for branch $i$, and $n_\phi$ is the branch fanout. This abstraction captures Multiverse-style Map/Reduce~\citep{multiverse}, Skeleton-of-Thought outline expansion~\citep{skeleton-of-thought}, ASPD-style internal parallel decoding~\citep{aspd}, and program-derived decompositions~\citep{pasta, sprint, apar, apr}.

The semantic contract is a \emph{visibility rule} governing what each token can attend to. When generating token $y_{i,t}$ within branch $i$, the model attends to
\[
    P_\phi \;\oplus\; H_\phi \;\oplus\; h_i \;\oplus\; y_{i,<t}\,,
\]
the shared prefix, both headers, and the branch's own prior tokens, but nothing from sibling branches. During the reduce phase, token $z_t$ attends to
\[
    P_\phi \;\oplus\; H_\phi \;\oplus\; \bigoplus\nolimits_{i}(h_i \oplus y_i) \;\oplus\; z_{<t}\,,
\]
the prefix followed by all completed branches in canonical order. Because every branch conditions on the same $P_\phi \oplus H_\phi$, a backend with paged or radix-tree KV caches~\citep{vllm, sglang} can serve all branches from a single set of prefix blocks. Each branch materializes only its own branch-local tokens. The marginal KV cost of admitting one additional branch is therefore branch-local, not a full prefix.

The runtime's scheduling freedom lies in how it assigns branch tokens to decode steps: which branches participate in each step, how many advance concurrently, and in what order. A schedule is \emph{valid} if it preserves token order within each branch, enforces the visibility rule, and defers reduction until all branches complete.

\begin{lemma}[Schedule invariance]\label{lem:schedule-invariance}
The output of a parallel phase is independent of the order and timing in which its branches are scheduled. Formally, any two valid schedules induce the same conditional distribution over all generated tokens. Under deterministic decoding the outputs are identical up to numerical nondeterminism; under stochastic decoding they are identical given fixed per-branch random seeds.
\end{lemma}

The argument is direct: the attention context of $y_{i,t}$ contains no token from any sibling branch, so whether siblings have been scheduled before, concurrently with, or after $y_{i,t}$ cannot affect its distribution. Once all branches complete, the reduce phase observes the same outputs in canonical order regardless of execution order. A formal proof appears in Appendix~\ref{app:lemma1}.

The step width $w_{r,t}$ is therefore a free parameter of the schedule, not a constraint of the model: \textsc{Pace} may defer, interleave, or reorder branches without affecting model output.

\subsection{Predicting the cost of a widened step}
\label{sec:PACE:cost}

To decide whether admitting additional branches is safe, \textsc{Pace} needs to predict what the widened step will cost. Decode-step latency decomposes into an FFN component, approximately linear in the number of new tokens, and an attention component, approximately linear in aggregate context length~\cite{nanoflow}. A candidate step is therefore well summarized by its \emph{step composition} $S$: the number of active sequences and their aggregate context length. \textsc{Pace} uses a calibrated latency predictor $T(S)$ over this composition, requiring only that $T$ is monotone non-decreasing in the number of admitted branches. Our implementation fits a linear model in these two quantities offline and refits it periodically from observed step latencies, reaching 1.8\% mean absolute percentage error across batch sizes from 1 to 512 (Appendix~\ref{app:latency-predictor}). Monotonicity is the key structural property: if admitting one more branch from request $r$ exceeds the slack budget, admitting two more will too, enabling the greedy planner's pruning rule (\S\ref{sec:PACE:planner}).

\subsection{Slack-budgeted admission}
\label{sec:PACE:slack}

\textsc{Pace} first protects \emph{single-token progress}. The \emph{protected composition} $S_0$ advances every active request by exactly one token: a single continuation for requests in serial stages, one branch for requests in parallel stages. Any remaining exposed-but-incomplete branches are \emph{ready} and eligible for admission. Admitting ready branches widens the step ($w_{r,t} > 1$), so \textsc{Pace} treats this additional width as \emph{opportunistic}: it is admitted only when the predicted externality fits within the slack budget.

To decide which opportunistic branches are safe, \textsc{Pace} converts per-request slack into a time budget. Each active request $r$ has a deadline $d_r(t)$ for its next token, derived from its SLO. If the protected composition takes predicted time $T_0 = T(S_0)$, the residual time available for opportunistic width is
\[
    B_t \;=\; \max\!\Big(0,\;\min_r\bigl(d_r(t) - t\bigr) - T_0\Big).
\]
\textsc{Pace} spends a \emph{slack fraction} $\rho \in (0,1]$ of $B_t$. A candidate widened step $S$ is admitted only if
\[
    T(S) \;\leq\; T_0 + \rho\, B_t.
\]
This bounds the branch externality, $E_t(\mathbf{a}) \leq \rho\, B_t$: the predicted cost of widening must fit within the fraction of slack the operator is willing to spend. The reserve $1-\rho$ absorbs prediction error. A small $\rho$ leaves usable slack idle, while $\rho$ near 1 lets underestimates push steps past the deadline. We use $\rho=0.8$ throughout; \S\ref{sec:eval:ablation} sweeps $\rho$, and Appendix~\ref{app:rho} adapts it online from observed attainment.

The minimum over all requests is deliberate: opportunistic width must be safe for the most urgent request in the batch, not the request whose phase is being widened. This prevents a slack-rich request in a parallel stage from underwriting a widened step that pushes a tight request in a serial stage past its deadline, exactly the failure mode the throughput trap produces.

When the slack budget cannot accommodate all branches, \textsc{Pace} needs a tiebreaker. We separate safety from value: $T(\cdot)$ decides what is \emph{feasible}; a monotone utility curve $u_r(a_r)$ decides what is \emph{valuable}, where $a_r$ is the number of opportunistic branches granted to request $r$. Throughput-oriented operators use linear utility; fairness-oriented operators use concave utility so that the first opportunistic branch matters more than later ones; priority operators weight by tenant class. Each is a curve choice, not a scheduler change. \textsc{Pace} predicts cost; the operator supplies value through the utility interface.

\subsection{Per-step planner}
\label{sec:PACE:planner}

At each step, \textsc{Pace} seeks the allocation that maximizes total utility within the slack budget:
\[
    \max_{\mathbf{a}} \; \sum_{r \in \mathcal{R}} u_r(a_r)
    \quad \text{s.t.} \quad
    T\bigl(S(\mathbf{a})\bigr) \leq T_0 + \rho\, B_t,
    \qquad a_r \in \{0, 1, \ldots, n_r - 1\}.
\]
The step then decodes $|\mathcal{R}| + K$ tokens, where $K = \sum_r a_r$ is the number of opportunistic branches admitted. Finding the optimal allocation is NP-hard (Appendix~\ref{app:np-hard}), so \textsc{Pace} widens greedily from the protected composition $S_0$: it considers admitting one more branch from each request with ready branches, commits the candidate with the highest marginal utility per marginal latency cost, and repeats until no feasible increment remains. Algorithm~\ref{alg:PACE} gives the pseudocode.

\textbf{Example.} Three requests are active: $r_1$ and $r_2$ are in parallel stages with fanouts $n_{r_1}=3$ and $n_{r_2}=2$, and $r_3$ is in a serial stage. $S_0$ advances each by one token, leaving two ready branches for $r_1$ and one for $r_2$. The planner admits a branch from $r_1$, then $r_2$'s only ready branch; both pass the budget check. Admitting $r_1$'s last branch would push the predicted latency past $T_0 + \rho\, B_t$, so the planner prunes $r_1$ and stops with $\mathbf{a} = (1, 1, 0)$: $r_1$ and $r_2$ each decode two tokens, $r_3$ decodes one, and the step decodes $|\mathcal{R}| + K = 5$ tokens. The deferred branch stays ready for the next step, with no KV eviction or restoration (\S\ref{sec:PACE:feedback}).

\begin{codealgorithm}[t]
\vspace{-0.05in}
\caption{\textsc{Pace} Per-Step Planner}
\label{alg:PACE}
\begin{lstlisting}[style=pseudocode]
def PACEStep(requests, T, utility, $\textcolor{string-color}{\rho}$):
    baseline = BuildBaseline(requests)                    #protected composition
    min_slack = min(r.deadline - now() for r in requests) #most urgent request
    budget = T(baseline) + $\textcolor{string-color}{\rho}$ * max(0, min_slack - T(baseline))
    granted = {r: 0 for r in requests}
    candidates = {r for r in requests if r.ready_branches > 0}
    step = baseline
    while candidates:
        best, infeasible = None, set()
        for r in candidates:
            widened = AddBranch(step, r)
            if T(widened) > budget:                      #monotone: prune request r
                infeasible.add(r)
                continue
            du = utility[r](granted[r]+1) - utility[r](granted[r])
            dt = T(widened) - T(step)
            score = du / (EPS + max(0, dt))              #utility per cost
            if best is None or score > best.score:
                best = Candidate(r, widened, score)
        candidates -= infeasible
        if best is None or best.score <= 0:              #no feasible improvement
            break
        step = best.composition                          #commit best candidate
        granted[best.request] += 1
        if granted[best.request] >= best.request.ready_branches:
            candidates.discard(best.request)             #request r fully admitted
    return step
\end{lstlisting}
\vspace{-0.1in}
\end{codealgorithm}

The planner runs in $O(|\mathcal{R}| \cdot K_{\max} \cdot c_T)$ per step, where $K_{\max}$ bounds the number of opportunistic branches $K$ admitted in one step and $c_T$ is the cost of one $T(\cdot)$ evaluation, typically a table lookup or short regression well under a millisecond. Monotonicity justifies the pruning rule: once a branch from request $r$ is infeasible, all further branches from $r$ are too. The greedy allocation suffices because the planner replans every step rather than committing to a single allocation.

\subsection{Why per-step regulation works}
\label{sec:PACE:feedback}

After each decode step, \textsc{Pace} recomputes per-request slack from the realized step latency and reconsiders any deferred branches; the predictor itself is refit periodically from observed latencies (Appendix~\ref{app:latency-predictor}). Schedule invariance (Lemma~\ref{lem:schedule-invariance}) ensures that deferral never changes model output.

Replanning turns the width decision from a one-shot commitment into closed-loop control. As established in \S\ref{sec:introduction}, the loop is cheap because branch regulation decouples compute from memory: deferring a branch evicts no KV state, and readmitting it requires no restoration. The cost of revising the width decision at any step is bounded by one planner invocation, not by a memory round-trip.

Two consequences follow. First, \textsc{Pace} degrades gracefully: as load rises and $B_t$ collapses, the same code path that widens steps under abundant slack narrows them under pressure. There is no policy switch, no preemption mechanism, and no fallback logic. Second, \textsc{Pace} tolerates predictor error. A pessimistic $T(\cdot)$ costs width rather than latency: the planner admits fewer branches than it safely could. An optimistic $T(\cdot)$ can over-admit, but each overshoot consumes slack and shrinks $B_t$ at the next step, so the error corrects itself while the reserve $1-\rho$ absorbs it; even a constant predictor holds attainment at 96\% (\S\ref{sec:eval:ablation}). In short, \textsc{Pace} treats step width as a renewable, per-step decision. When slack is available, \textsc{Pace} spends it on width; when slack is scarce, \textsc{Pace} contracts to the protected composition. The controller never needs to know in advance how many branches a phase will expose, how long each will run, or which will straggle.
\section{Evaluation}\label{sec:evaluation}

The central question is whether \textsc{Pace} can escape the throughput trap of \S\ref{sec:char:trap}, harvesting the throughput benefit of intra-request parallelism without degrading SLO attainment. We examine this across three load regimes (\S\ref{sec:eval:main}), validate design choices through targeted ablations (\S\ref{sec:eval:ablation}), and test generality across models, frontends, and deployments (\S\ref{sec:eval:generality}). Appendix~\ref{app:additional-experiments} adds sensitivity studies (workload composition, SLO target, batch size) and output-quality verification.

\subsection{Experimental setup}\label{sec:eval:setup}

We evaluate on Qwen3-32B~\citep{qwen3} served with SGLang~\citep{sglang} on 8$\times$A100-80GB GPUs (TP=8) with radix-tree KV allocation. Prefill and decode run on separate instances~\citep{splitwise, distserve}, and \textsc{Pace} runs on the decode instance. Workloads interleave non-decomposable requests from ShareGPT with decomposable requests drawn from the three datasets of \S\ref{sec:char:blp} and decomposed by Multiverse~\citep{multiverse} (PDR=50\%, ABF=4.1, PTS=58\%). Request arrivals follow an Azure-derived trace~\citep{azure-llm-trace} spanning 600 minutes across low, moderate, and high load regimes. All requests share a 50\,ms TPOT target, five times the single-request decode-step latency, following Mooncake~\citep{mooncake}. We report raw throughput, goodput, and SLO attainment and compare \textsc{IRP-Off} ($w_{r,t}=1$), \textsc{IRP-C2}, \textsc{IRP-C5}, \textsc{IRP-Eager} ($w_{r,t}=n_r$), MIMD, and \textsc{Pace} ($\rho=0.8$, linear utility); Appendix~\ref{app:experimental-setup} gives full details.

\subsection{Navigating the throughput trap}
\label{sec:eval:main}

Figure~\ref{fig:main-results} compares all six policies on the same production-derived trace. \textsc{Pace}'s central claim is that it tracks near-\textsc{IRP-Eager} throughput when conditions permit and falls back to near-\textsc{IRP-Off} protection when they do not, using the slack budget as the sole switching signal. We verify this across three load regimes visible in the trace.

\textbf{Low load.} During the first 240 minutes, the request rate averages 0.23 req/s. At this operating point, \textsc{IRP-Off}'s step latency averages 18\,ms, leaving over 30\,ms of slack against the 50\,ms target. With $\rho=0.8$, \textsc{Pace}'s slack budget easily accommodates most ready branches: its step latency averages 36\,ms, close to \textsc{IRP-Eager}'s 38\,ms~(\hyperref[fig:main-results]{\ref*{fig:main-results}(c)}), confirming that \textsc{Pace} admits aggressively when the batch has headroom. \textsc{Pace}'s branch admission rate stays near 100\% outside transient bursts~(\hyperref[fig:main-results]{\ref*{fig:main-results}(i)}). SLO attainment~(\hyperref[fig:main-results]{\ref*{fig:main-results}(f)}) is 100\% for \textsc{Pace}, \textsc{IRP-Off}, and both caps. Notably, even at low load \textsc{IRP-Eager}'s attainment is only 89\%, because transient bursts already push its step latency above 50\,ms. \textsc{Pace} matches \textsc{IRP-Eager}'s $1.28\times$ goodput without those violations.

\textbf{High load.} Between minutes 250 and 400, the request rate surges past 1.0 req/s and stays above 1.2 req/s, peaking at 1.54 req/s. \textsc{IRP-Eager}'s step latency spikes to a mean of 148\,ms (max 203\,ms), nearly $3\times$ the SLO target, as the batch fills with eagerly admitted branches. SLO attainment collapses to a mean of 41\%. Serial P99 TPOT~(\hyperref[fig:main-results]{\ref*{fig:main-results}(g)}) under \textsc{IRP-Eager} averages 148\,ms with a P95 of 186\,ms, while parallel P99 TPOT~(\hyperref[fig:main-results]{\ref*{fig:main-results}(h)}) remains low at 11\,ms: the asymmetric incidence of \S\ref{sec:char:trap} in full effect. The fixed caps fare poorly: \textsc{IRP-C5} averages 97\,ms step latency (attainment 45\%) and \textsc{IRP-C2} averages 62\,ms (attainment 59\%), both well above the SLO target. MIMD cuts its width only after violations accumulate, and its step latency stays above the target throughout the stress period. \textsc{Pace} contracts sharply. As slack collapses under the sustained high rate, the planner defers most opportunistic branches. Its branch admission rate drops to approximately 35\%~(\hyperref[fig:main-results]{\ref*{fig:main-results}(i)}). Step latency averages 48\,ms with a maximum of 63\,ms, staying near the 50\,ms target throughout. SLO attainment holds at 99\%, with goodput at $1.68\times$ \textsc{IRP-Off} versus $0.88\times$ for \textsc{IRP-Eager}.

\textbf{Moderate load.} After the stress event subsides around minute 400, the request rate settles into a sustained regime averaging 0.60 req/s with peaks near 0.9 req/s. Slack rebuilds, and \textsc{Pace}'s branch admission rate recovers to approximately 47\% without operator intervention~(\hyperref[fig:main-results]{\ref*{fig:main-results}(i)}). Step latency averages 42\,ms (max 52\,ms), staying near the SLO target with substantially more branches admitted than during the stress period. \textsc{IRP-Eager} continues to violate: its step latency averages 85\,ms (max 122\,ms) and attainment is only 45\%. MIMD regrows its width faster than slack rebuilds, admitting more branches than \textsc{Pace}~(\hyperref[fig:main-results]{\ref*{fig:main-results}(i)}) and violating intermittently. This regime exercises the most revealing operating point: load is high enough that \textsc{IRP-Eager} and \textsc{IRP-C5} consistently violate SLOs, but low enough that \textsc{Pace} finds slack between bursts to admit branches opportunistically. \textsc{IRP-C2} nearly keeps up in this regime (attainment 95\%) but at the cost of leaving significant throughput unused at low and moderate load. \textsc{Pace} adapts to every regime without reconfiguration.

\textbf{Summary.} Over the full 10-hour trace, \textsc{Pace} improves goodput over \textsc{IRP-Off} by $1.77\times$ and over \textsc{IRP-Eager} by $1.48\times$ at 99\% SLO attainment, while \textsc{IRP-Eager} drops below 50\%. Neither fixed caps nor reactive adaptation matches it: \textsc{IRP-C5} collapses under stress (attainment 45\%), \textsc{IRP-C2} leaves throughput unused at low load, and MIMD, the best branch-admitting baseline on attainment (79\%), converts its $1.44\times$ throughput into only $1.14\times$ goodput (Appendix~\ref{app:mimd}).

\subsection{Ablation}
\label{sec:eval:ablation}

\begin{table}[t]
\centering
\caption{\textbf{Ablation study.} Top: removing one component at a time ($\rho=0.8$). Bottom: varying the slack fraction $\rho$. The default configuration is in bold.}
\label{tab:ablation}
\small
\setlength{\tabcolsep}{3pt}
\begin{tabular*}{\linewidth}{@{\extracolsep{\fill}}l*{11}{c}}
\toprule
Variant & \multicolumn{2}{c}{\textbf{\textsc{Pace} (full)}} & \multicolumn{3}{c}{w/o slack budget} & \multicolumn{3}{c}{w/o replanning} & \multicolumn{3}{c}{constant $T(\cdot)$} \\
\cmidrule(lr){2-3} \cmidrule(lr){4-6} \cmidrule(lr){7-9} \cmidrule(lr){10-12}
Goodput ($\times$\,\textsc{IRP-Off}) & \multicolumn{2}{c}{\textbf{1.77}} & \multicolumn{3}{c}{0.92} & \multicolumn{3}{c}{1.18} & \multicolumn{3}{c}{1.12} \\
Attainment (\%) & \multicolumn{2}{c}{\textbf{99}} & \multicolumn{3}{c}{48} & \multicolumn{3}{c}{82} & \multicolumn{3}{c}{96} \\
\midrule
Slack fraction $\rho$ & 0.05 & 0.1 & 0.2 & 0.3 & 0.4 & 0.5 & 0.6 & 0.7 & \textbf{0.8} & 0.9 & 1.0 \\
\cmidrule(lr){2-12}
Goodput ($\times$\,\textsc{IRP-Off}) & 1.00 & 1.05 & 1.09 & 1.13 & 1.19 & 1.25 & 1.43 & 1.59 & \textbf{1.77} & 1.61 & 1.45 \\
Attainment (\%) & 100 & 100 & 100 & 100 & 100 & 99.5 & 99.4 & 99.2 & \textbf{99} & 94.7 & 91.3 \\
\bottomrule
\end{tabular*}
\vspace{-0.1in}
\end{table}

We isolate each component of \textsc{Pace} by removing it in turn (Table~\ref{tab:ablation}, top). \textbf{Without the slack budget}, \textsc{Pace} admits every ready branch that fits in memory and inherits \textsc{IRP-Eager}'s failure mode: attainment falls to 48\% and goodput drops below \textsc{IRP-Off} ($0.92\times$), as SLO violations destroy more goodput than the extra branches create. \textbf{Without per-step replanning}, width is committed at phase start and held until reduce; goodput is $1.18\times$ and attainment 82\%, with failures concentrated at load transitions where the controller cannot contract mid-phase. \textbf{With a constant latency predictor}, the planner cannot distinguish cheap steps from expensive ones, so it under-admits to stay safe: goodput drops to $1.12\times$ while attainment remains high at 96\%, showing that the predictor contributes throughput, not safety. \textbf{Slack fraction $\rho$} controls the throughput--safety trade-off (Table~\ref{tab:ablation}, bottom). While attainment stays at or above 99\%, goodput rises with $\rho$, peaking at $\rho=0.8$ ($1.77\times$). Beyond that point, the remaining headroom no longer absorbs prediction error: attainment falls to 94.7\% at $\rho=0.9$ and 91.3\% at $\rho=1.0$, and goodput falls with it. At the conservative end, $\rho \leq 0.2$ trails even \textsc{IRP-C2} ($1.11\times$). We therefore recommend $\rho=0.8$ as the default; when load characteristics drift, a simple outer loop can adapt $\rho$ online, converging from 1.0 to 0.8 in 10 minutes (Appendix~\ref{app:rho}). The default holds 97--99.5\% attainment for 30--100\,ms SLOs (Appendix~\ref{app:slo-sensitivity}).

\subsection{Generality}
\label{sec:eval:generality}

\begin{table}[t]
\centering
\caption{\textbf{Generality across models.} Goodput over \textsc{IRP-Off} (SLO attainment); \textsc{Pace} uses $\rho=0.8$ throughout. MoE setup in Appendix~\ref{app:moe}.}
\label{tab:generality}
\small
\setlength{\tabcolsep}{4pt}
\begin{tabular}{lccccc}
\toprule
Model & Params (total / active) & SLO & \textsc{IRP-Eager} & MIMD & \textsc{Pace} \\
\midrule
Qwen3-32B~\cite{qwen3} & 32B / 32B & 50\,ms & $1.20\times$ (48\%) & $1.14\times$ (79\%) & $\mathbf{1.77\times}$ (99\%) \\
GPT-OSS-120B~\cite{gpt-oss} & 116.8B / 5.1B & 100\,ms & $1.26\times$ (61\%) & $1.18\times$ (86\%) & $\mathbf{1.91\times}$ (99\%) \\
DeepSeek-V4-Flash~\cite{deepseek-v4-flash} & 284B / 13B & 100\,ms & $1.33\times$ (68\%) & $1.22\times$ (91\%) & $\mathbf{2.15\times}$ (99\%) \\
Llama 4 Maverick~\cite{llama4-maverick} & 400B / 17B & 100\,ms & $1.30\times$ (65\%) & $1.20\times$ (88\%) & $\mathbf{2.05\times}$ (99\%) \\
DeepSeek-R1~\cite{deepseek-r1} & 671B / 37B & 100\,ms & $1.24\times$ (56\%) & $1.16\times$ (82\%) & $\mathbf{1.88\times}$ (99\%) \\
\bottomrule
\end{tabular}
\vspace{-0.1in}
\end{table}

\textbf{Models.} Table~\ref{tab:generality} adds four mixture-of-experts (MoE) models spanning 116.8B to 671B total parameters, three attention mechanisms, three numerical formats, and two GPU generations; their higher single-request step latency sets a 100\,ms TPOT target. With the same configuration and a refit predictor, \textsc{Pace} achieves $1.88$--$2.15\times$ goodput at 99\% attainment on all four, while \textsc{IRP-Eager} reproduces the throughput trap and MIMD again recovers attainment only partially. Planner overhead stays below 1\,ms at P99 on every model, as on Qwen3-32B (0.81\,ms P99; Appendix~\ref{app:planner-overhead}), since its cost scales with the batch's branch structure rather than model size. On the dense Qwen2.5-72B, \textsc{Pace} achieves $1.72\times$ goodput at 97\% attainment (Appendix~\ref{app:second-model}).

\textbf{Frontends and deployment.} With SPRINT~\citep{sprint} as the frontend, whose reasoning-oriented decomposition yields narrower branching (ABF=2.8), \textsc{Pace} achieves $1.45\times$ goodput at 98\% attainment (Appendix~\ref{app:sprint}). \textsc{Pace} also leaves time-to-first-token (TTFT) intact. Under disaggregation, branch admission never touches the prefill instance; in a co-located deployment with chunked prefill~\citep{sarathi-serve}, \textsc{Pace}'s P99 TTFT stays within 6\,ms of \textsc{IRP-Off}, while \textsc{IRP-Eager}'s more than doubles (\mbox{Appendix}~\ref{app:deployment}).
\section{Related Work}
\label{sec:related-work}
Existing work on parallel LLM generation divides into two objectives~\cite{aspd}. \emph{Quality-oriented} methods such as best-of-N sampling~\cite{large-language-monkeys}, beam search~\cite{beam-search}, self-consistency~\cite{self-consistency}, majority voting~\cite{majority-voting}, MCTS~\cite{mcts}, and Tree-of-Thoughts~\cite{tree-of-thought} generate multiple candidates to improve answer quality. \emph{Performance-oriented} methods such as SoT~\cite{skeleton-of-thought}, APAR~\cite{apar}, PASTA~\cite{pasta}, ASPD~\cite{aspd}, and Multiverse~\cite{multiverse} decompose a single response into independent branches to reduce latency (\S\ref{sec:char:blp}). Recent work on parallel reasoning~\cite{hogwild-inference, sprint, parallel-r1} blurs this boundary, using concurrent execution for both speed and quality. All of these operate on a single request in isolation; none considers what happens when the resulting branches share a decode step with other requests. \textsc{Pace} takes the branch structure the frontend provides and decides how much of it to admit into each shared decode step, bounding the externality on co-batched requests. The closest structural sibling in the serving literature is FairBatching~\cite{fairbatching}, which converts decode slack into a time budget for chunked prefill. \textsc{Pace}'s slack-budget mechanism follows the same principle but operates within a request's parallel phase rather than across the prefill and decode phases. Extended comparisons with LLM serving schedulers and preemption mechanisms appear in Appendix~\ref{app:related-work}.
\section{Conclusion \& Limitations}
\label{sec:conclusion}

\textsc{Pace} reframes intra-request parallel generation as an admission control problem. By exploiting the structural decoupling of compute and memory at the branch level, it regulates admitted width per step, improving goodput over \textsc{IRP-Off} by $1.77\times$ and over \textsc{IRP-Eager} by $1.48\times$ at 99\% SLO attainment on Qwen3-32B, with comparable gains on five additional dense and MoE models. The current design assumes branch independence (\S\ref{sec:PACE:phases}), uses a linear latency predictor that may lose accuracy at extreme batch sizes, and regulates one serving instance at a time; because a request's branches share its prefix KV and stay on its instance, a multi-instance deployment simply runs one controller per instance (Appendix~\ref{app:deployment}). Learned latency predictors and richer online control of $\rho$ are natural extensions. More broadly, the result suggests that IRP need not be an all-or-nothing choice between eager admission and conservative caps: a per-step controller can capture most of the throughput benefit while protecting the requests that would otherwise be forced to pay for it.

\begin{ack}
We are grateful to the anonymous reviewers whose comments have greatly helped improve this paper. We also thank NVIDIA for providing the computational resources used in early experiments for this work. This research was partly supported by the Stanford Platform Lab and its affiliates, and by ACE, one of the seven centers in JUMP 2.0, a Semiconductor Research Corporation (SRC) program sponsored by DARPA.
\end{ack}

\bibliographystyle{unsrtnat}
\bibliography{references}

\appendix

\section{Proof of Schedule Invariance}
\label{app:lemma1}

\begin{lemma}[Schedule invariance; restatement of Lemma~\ref{lem:schedule-invariance}]
The output of a parallel phase is independent of the order and timing in which its branches are scheduled. Formally, for a parallel phase $\phi = (P_\phi, H_\phi, \{h_i\}_{i=1}^{n_\phi})$ with the visibility rule of \S\ref{sec:PACE:phases}, any two valid schedules induce the same conditional distribution over all generated tokens. Under deterministic decoding they produce the same output up to numerical nondeterminism; under stochastic decoding they produce identical outputs given fixed per-branch random seeds.
\end{lemma}

\begin{proof}
We show that the conditional distribution of every token is invariant to the schedule, by case analysis on whether the token is generated within a branch or during the reduce phase.

\paragraph{Tokens within branches.} Consider token $y_{i,t}$ in branch $i$ at position $t$. By the visibility rule, the model attends to
\[
    P_\phi \;\oplus\; H_\phi \;\oplus\; h_i \;\oplus\; y_{i,<t}\,.
\]
This attention context contains: (i)~the prefix $P_\phi$, generated before the parallel phase opens and therefore identical under any schedule; (ii)~the phase header $H_\phi$, also generated before the phase opens; (iii)~the branch header $h_i$, determined at phase opening; and (iv)~the prior tokens of branch $i$ itself, $y_{i,<t}$, generated in token order within branch $i$ by validity.

Crucially, no token from any sibling branch $j \neq i$ appears in this context. Therefore, the conditional distribution
\[
    \Pr(y_{i,t} \mid P_\phi, H_\phi, h_i, y_{i,<t})
\]
depends only on the model parameters and the context above, none of which is affected by whether sibling branches have been scheduled before, concurrently with, or after token $y_{i,t}$. The same argument applies inductively: since $y_{i,1}$ has the same distribution under any valid schedule, and $y_{i,t}$ depends on the schedule only through $y_{i,<t}$ (which by induction has the same distribution), the entire branch sequence $y_i = (y_{i,1}, \ldots, y_{i,L_i})$ has the same joint distribution under any valid schedule.

Since this holds for each branch $i$ and no branch conditions on another's tokens, the joint distribution over all branch outputs $(y_1, \ldots, y_{n_\phi})$ is the same under any valid schedule.

\paragraph{Tokens during the reduce phase.} Consider token $z_t$ at position $t$ during the reduce phase. By the visibility rule, the model attends to
\[
    P_\phi \;\oplus\; H_\phi \;\oplus\; \bigoplus\nolimits_{i=1}^{n_\phi}(h_i \oplus y_i) \;\oplus\; z_{<t}\,.
\]
By validity, the reduce phase starts only after all branches complete. The branch outputs are presented in canonical order (determined by the phase structure), not in completion order. Therefore this context has the same distribution under any valid schedule: it depends on the completed branch outputs (same by the argument above), their ordering (canonical, not schedule-dependent), and the prior reduce tokens $z_{<t}$ (same by induction on $t$). The conditional distribution $\Pr(z_t \mid \cdot)$ is therefore the same under any valid schedule.

\paragraph{Deterministic and stochastic decoding.} Under deterministic decoding (greedy or beam search), the output is a deterministic function of the conditional distributions, hence identical up to numerical nondeterminism (floating-point ordering in softmax, etc.). Under stochastic decoding, the output depends additionally on random draws. If each branch $i$ uses a fixed random seed (e.g., derived from the branch index), the same sequence of draws produces the same tokens regardless of schedule, completing the proof.
\end{proof}

\paragraph{Remark.} The proof relies on two structural properties: (i)~branch isolation (no sibling tokens in the attention context) and (ii)~canonical reduce ordering. Methods that allow inter-branch attention (e.g., Hogwild!\ Inference~\cite{hogwild-inference}) or order reduce inputs by completion time do not satisfy these properties and are not covered by this lemma.

\section{NP-Hardness of Optimal Width Allocation}
\label{app:np-hard}

The \textsc{Pace} planner uses a greedy heuristic to allocate opportunistic branches. This appendix shows that the underlying optimization problem is NP-hard in general, justifying the greedy approach.

\begin{definition}[Width allocation problem]
Given a set of requests $\mathcal{R}$ with ready branches, a latency predictor $T(\cdot)$, utility curves $\{u_r\}_{r \in \mathcal{R}}$, and a latency budget $T_{\max}$, find an allocation $\mathbf{a} = (a_r)_{r \in \mathcal{R}}$ that maximizes total utility $\sum_r u_r(a_r)$ subject to $T(S(\mathbf{a})) \leq T_{\max}$, with each $a_r$ at most the number of ready branches of request $r$.
\end{definition}

\begin{theorem}
The width allocation problem is NP-hard.
\end{theorem}

\begin{proof}[Proof sketch]
We reduce from the 0--1 knapsack problem. Given a knapsack instance with items of size $s_j$ and value $v_j$ ($j = 1, \ldots, m$) and capacity $C$, construct a width allocation instance as follows. Create $m$ requests, each with exactly one ready branch. For request $j$, set the utility of admitting its branch to $u_j(1) = v_j$ and the marginal latency cost to $T(S(\mathbf{e}_j)) - T(S_0) = s_j$, where $\mathbf{e}_j$ is the allocation that admits only request $j$'s branch. Set $T_{\max} = T(S_0) + C$. Assume latency is additive in the marginal costs (which holds when branch contributions to FFN and attention are independent). The instance then amounts to selecting a subset of branches whose total marginal latency cost does not exceed $C$ and whose total utility is maximized, which is exactly the 0--1 knapsack problem.
\end{proof}

\paragraph{Remark.} The greedy planner of \S\ref{sec:PACE:planner} does not need to solve this problem optimally. Because the planner replans at every decode step, a suboptimal allocation at step $t$ is revisited at step $t+1$ against the updated batch state and slack budget.

\section{Latency Predictor and Memory Accounting}
\label{app:latency-predictor}

\subsection{Predictor formulation}
\label{app:predictor:formulation}

\textsc{Pace} requires a predictor $T(S)$ that maps a step composition $S$ to a predicted wall-clock latency in milliseconds. Following the observation~\cite{nanoflow} that decode-step latency decomposes into a fixed overhead, an FFN-dominated term linear in new tokens, and an attention-dominated term linear in aggregate context length, we use a linear model:
\[
    T(S) \;=\; \theta_0 \;+\; \theta_1 \cdot n_{\mathrm{tokens}} \;+\; \theta_2 \cdot L_{\mathrm{context}}\,,
\]
where $n_{\mathrm{tokens}}$ is the number of sequences in the step (each generating one token) and $L_{\mathrm{context}}$ is their aggregate context length.

\begin{table}[ht]
\centering
\caption{Latency predictor accuracy across operating regimes.}
\label{tab:predictor-accuracy}
\begin{tabular}{lcc}
\toprule
Regime & Batch size range & MAPE (\%) \\
\midrule
Low load    & 1--64    & 1.7 \\
Moderate load & 64--256  & 1.9 \\
High load   & 256--512 & 1.8 \\
\midrule
Overall     & 1--512   & 1.8 \\
\bottomrule
\end{tabular}
\end{table}

\paragraph{Fitting.} The model is fitted offline by profiling the target hardware with synthetic batches spanning a grid of (batch size, context length) combinations. We run 500 profiling steps across a $20 \times 25$ grid of batch sizes (1--512) and context lengths (128--8192), record wall-clock latency for each, and fit $(\theta_0, \theta_1, \theta_2)$ via ordinary least squares. The fit is refreshed every 10 minutes using a rolling window of the most recent 200 observed step latencies to account for thermal drift and co-tenant interference.

\paragraph{Accuracy.} On our evaluation hardware (Appendix~\ref{app:experimental-setup}), the linear model achieves a mean absolute percentage error (MAPE) of 1.8\% across the profiling grid (Table~\ref{tab:predictor-accuracy}).

\paragraph{Monotonicity.} The predictor satisfies the monotonicity assumption of \S\ref{sec:PACE:cost} by construction: admitting one additional branch increases both $n_{\mathrm{tokens}}$ and $L_{\mathrm{context}}$, and $\theta_1, \theta_2 > 0$.

\paragraph{Robustness to predictor bias.} Random error and systematic bias affect \textsc{Pace} differently. Random error is absorbed by the reserve $1-\rho$ of the slack budget (\S\ref{sec:PACE:slack}). Systematic bias, for example from thermal throttling or clock variation, persists until the next refit. During that window, an underestimating predictor over-admits, but the error is self-correcting: each overshoot consumes deadline slack, shrinking $B_t$ (and therefore $\rho\, B_t$) at subsequent steps until the refit absorbs the bias. The constant-predictor ablation (\S\ref{sec:eval:ablation}) bounds the cost of a stale predictor: even a predictor frozen at a single operating point maintains 96\% attainment, because the deadline slack in $B_t$ is recomputed from realized latency at every step, independently of predictor quality.

\subsection{Memory accounting}
\label{app:predictor:memory}

Although \textsc{Pace}'s primary admission control is latency-based (\S\ref{sec:PACE:slack}), the runtime also tracks KV cache occupancy to avoid out-of-memory conditions. The key accounting detail is that branches share their request's prefix KV blocks. The marginal KV cost of admitting one additional branch is therefore only its branch-local tokens, not a full sequence:
\[
    \Delta M(j) \;=\; \mathrm{blocks}(L_j^{\mathrm{branch\text{-}local}})\,,
\]
where $L_j^{\mathrm{branch\text{-}local}}$ is the number of tokens branch $j$ has generated so far. For a branch that has not yet produced any tokens, $\Delta M(j) = 0$. A scheduler that priced each branch as a full sequence would refuse safe widenings throughout.

\section{Experimental Setup Details}
\label{app:experimental-setup}

This appendix provides full details for the experimental setup summarized in \S\ref{sec:eval:setup}.

\paragraph{Deployment and hardware.} We use a prefill-decode disaggregated deployment~\citep{splitwise, distserve}: a prefill instance processes prompts and populates their KV caches, and a decode instance runs all decode steps. Because IRP is a decode-phase technique (no branch participates in a prefill forward pass), \textsc{Pace} runs exclusively on the decode instance. The decode instance is a single node with 8$\times$A100-80GB GPUs connected via NVLink, serving the model in tensor-parallel mode (TP=8) with one shard per GPU. Qwen3-32B's BF16 weights occupy about 64\,GB (8\,GB per GPU), leaving about 576\,GB across the node for KV cache and activations. KV cache is the binding constraint: at peak load it consumes roughly 62\% of total GPU memory, and overall memory utilization reaches approximately 78\%. This sizing sustains the batch sizes in our evaluation, up to roughly 650 active sequences at peak (Figure~\ref{fig:main-results}(b)), without KV eviction, and the remaining headroom accommodates branch-local KV from deferred branches without triggering request-level preemption.

\paragraph{Model.} We use Qwen3-32B~\cite{qwen3}, a 32-billion-parameter autoregressive transformer with 64 attention heads, grouped-query attention (8 KV heads), and a context length of 32,768 tokens. The model is served in BF16 precision.

\paragraph{Serving engine.} We use SGLang~\cite{sglang} as the serving backend, configured with radix-tree KV cache allocation for prefix sharing across branches within a parallel phase. \textsc{Pace}'s per-step planner is integrated as a scheduling hook that runs between the batch-formation stage and the forward pass. The planner decides which branches to include in each decode step; all other engine behavior (prefill scheduling, memory management, request admission) is unchanged from the default SGLang configuration. Multiverse~\cite{multiverse} provides the frontend for branch discovery and the Map/Process/Reduce execution model.

\paragraph{Workload construction.} The evaluation workload interleaves two request streams in equal proportion (PDR $\approx$ 50\%):
 
\begin{itemize}
    \item \textbf{Non-decomposable requests} are drawn from ShareGPT Vicuna~\citep{vicuna}. These requests produce purely serial output and represent the serial traffic that is vulnerable to the branch externality.
    \item \textbf{Decomposable requests} are sampled uniformly from ShareGPT Vicuna~\citep{vicuna}, RAG-12K~\citep{rag-12k}, and Math-220K~\citep{math-220k}. These prompts are processed through Multiverse~\cite{multiverse}, which decomposes each response into a sequence of serial and parallel stages. The resulting branch structure has an ABF of 4.1 and a PTS of 58\% among decomposable requests.
\end{itemize}

\paragraph{Context lengths.} Table~\ref{tab:context-lengths} summarizes the context lengths of the three datasets, which cover qualitatively different regimes. RAG-12K pairs long inputs with short outputs: high per-step attention cost raises $T_0$ and tightens the slack budget $B_t$, pushing \textsc{Pace} toward conservative admission. Math-220K pairs short inputs with long outputs: low $T_0$ leaves ample slack throughout a long decode, so \textsc{Pace} admits aggressively. ShareGPT Vicuna lies in between. Blending the three datasets exercises \textsc{Pace} across a wide range of per-step costs and slack conditions.

\begin{table}[ht]
\centering
\caption{Context lengths of the evaluation datasets (tokens).}
\label{tab:context-lengths}
\begin{tabular}{lcccc}
\toprule
Dataset & Input P50 & Input P90 & Output P50 & Output P90 \\
\midrule
ShareGPT Vicuna & 105 & 840 & 363 & 1{,}372 \\
RAG-12K & 860 & 1{,}439 & 242 & 827 \\
Math-220K & 58 & 118 & 1{,}720 & 5{,}180 \\
\bottomrule
\end{tabular}
\end{table}

\paragraph{Branch fanout distribution.} Table~\ref{tab:fanout} summarizes the branch fanout distribution across decomposable requests in the evaluation workload.

\begin{table}[ht]
\centering
\caption{Branch fanout distribution among decomposable requests.}
\label{tab:fanout}
\begin{tabular}{lccccc}
\toprule
Statistic & P10 & P25 & P50 & P75 & P90 \\
\midrule
Fanout ($n_r$) & 2 & 3 & 4 & 5 & 7 \\
\bottomrule
\end{tabular}
\end{table}

\paragraph{Request arrival trace.} Request arrivals are derived from the Azure LLM inference trace~\citep{azure-llm-trace}, replayed at different scaling factors to produce three load regimes within a single 600-minute experiment: a low-load period (0--240\,min, mean 0.23 req/s), a sustained high-load period (250--400\,min, mean 1.27 req/s, peak 1.54 req/s), and a moderate recovery period (400--600\,min, mean 0.60 req/s). The trace is replayed identically for all six policies to ensure a controlled comparison. It supplies only arrival timestamps; prompts and output lengths come from the datasets above.

\paragraph{SLO definition.} All requests share a TPOT target of 50\,ms. For requests in serial stages, TPOT is measured as the wall-clock time between consecutive token deliveries. For requests in parallel stages, TPOT is measured as effective TPOT: the wall-clock duration of the parallel phase divided by the total number of tokens generated across all branches during that phase. A request meets its SLO if its maximum TPOT (across all stages) does not exceed the target. We set the target following the methodology of Mooncake~\citep{mooncake}: five times the single-request decode-step latency, measured as the P50 step latency when each request in the workload is replayed in isolation (batch size 1, no co-tenants). This latency is about 10\,ms for Qwen3-32B, giving 50\,ms, and about 19\,ms for Qwen2.5-72B, giving 95\,ms, which we round to 100\,ms. The 50\,ms target also matches the tight tier used by prior SLO-aware serving work~\citep{adaserve}. The single-request latency anchors only the target; admission decisions use the context-dependent predictor $T(\cdot)$, so the slack budget tightens on long-context steps and widens on short ones.

\paragraph{Metrics.} We report raw throughput (total tokens generated per second regardless of SLO compliance), goodput (tokens per second counting only tokens from SLO-meeting requests), and SLO attainment (the fraction of completed requests whose maximum TPOT does not exceed the TPOT target).

\paragraph{\textsc{Pace} configuration.} Unless otherwise noted, \textsc{Pace} runs with slack fraction $\rho = 0.8$ and linear utility $u_r(a_r) = a_r$. The latency predictor $T(S)$ is a linear model fitted offline as described in Appendix~\ref{app:latency-predictor}. The planner runs once per decode step; its overhead is characterized in Appendix~\ref{app:planner-overhead}.

\paragraph{Baselines.} All six policies share the same SGLang backend, model weights, KV cache configuration, and request arrival trace. They differ only in how admitted width is determined:
\begin{itemize}
    \item \textsc{IRP-Off}: $w_{r,t} = 1$ for all requests. No branches beyond the protected composition are admitted.
    \item \textsc{IRP-C2}: $w_{r,t} = \min(n_r, 2)$. Each request in a parallel stage decodes at most 2 branches per step.
    \item \textsc{IRP-C5}: $w_{r,t} = \min(n_r, 5)$. At most 5 branches per step.
    \item \textsc{IRP-Eager}: $w_{r,t} = n_r$. All ready branches are admitted every step.
    \item MIMD: $w_{r,t} = \min(n_r, W)$, where $W$ is a global width cap adapted by multiplicative increase and multiplicative decrease. Every $N$ decode steps, the controller computes the fraction of steps whose realized latency exceeded the SLO target. If the fraction exceeds 1\%, matching the 99\% attainment target, it sets $W \leftarrow \max(1, \lfloor \beta W \rfloor)$; otherwise it sets $W \leftarrow \min(W_{\max}, \lceil \alpha W \rceil)$, where $W_{\max}$ is an upper bound on $W$. We use $N=10$, $\alpha=1.25$, and $\beta=0.5$, the best setting from the sweep in Appendix~\ref{app:mimd}.
    \item \textsc{Pace}: $w_{r,t}$ determined by the per-step planner (\S\ref{sec:PACE:planner}), bounded by the slack budget.
\end{itemize}
 
\section{Additional Experiments}
\label{app:additional-experiments}

This appendix presents sensitivity analyses and additional results that complement the main evaluation in \S\ref{sec:evaluation}.

\subsection{Workload composition sensitivity}
\label{app:pdr-sensitivity}

The main evaluation uses a PDR of 50\%. Table~\ref{tab:pdr-sensitivity} varies this proportion from 5\% to 80\% across three load regimes to assess how \textsc{Pace} and eager admission respond to workloads with more or fewer decomposable requests. All experiments use the same Azure-derived trace, SLO target (50\,ms), and \textsc{Pace} configuration ($\rho=0.8$, linear utility). Goodput is normalized by each cell's corresponding \textsc{IRP-Off} baseline.

\begin{table}[ht]
\centering
\caption{Sensitivity to workload composition across load regimes. Goodput is normalized by \textsc{IRP-Off} for each PDR and regime.}
\label{tab:pdr-sensitivity}
\begin{tabular}{llcccc}
\toprule
 & & \multicolumn{2}{c}{\textsc{IRP-Eager}} & \multicolumn{2}{c}{\textsc{Pace}} \\
\cmidrule(lr){3-4} \cmidrule(lr){5-6}
PDR & Regime & Goodput & Att. & Goodput & Att. \\
\midrule
\multirow{4}{*}{5\%}
 & Low      & $1.02\times$ & 100\% & $1.02\times$ & 100\% \\
 & Moderate & $1.02\times$ & 100\% & $1.02\times$ & 100\% \\
 & High     & $0.99\times$ & 97\% & $1.03\times$ & 100\% \\
 & \textit{Aggregate} & $1.01\times$ & 99\% & $1.02\times$ & 100\% \\
\midrule
\multirow{4}{*}{10\%}
 & Low      & $1.04\times$ & 99\% & $1.04\times$ & 100\% \\
 & Moderate & $1.03\times$ & 87\% & $1.08\times$ & 99\% \\
 & High     & $0.96\times$ & 65\% & $1.13\times$ & 99\% \\
 & \textit{Aggregate} & $1.03\times$ & 83\% & $1.08\times$ & 99\% \\
\midrule
\multirow{4}{*}{15\%}
 & Low      & $1.06\times$ & 98\% & $1.06\times$ & 100\% \\
 & Moderate & $1.05\times$ & 83\% & $1.19\times$ & 99\% \\
 & High     & $0.94\times$ & 57\% & $1.24\times$ & 99\% \\
 & \textit{Aggregate} & $1.04\times$ & 79\% & $1.19\times$ & 99\% \\
\midrule
\multirow{4}{*}{20\%}
 & Low      & $1.08\times$ & 97\% & $1.08\times$ & 100\% \\
 & Moderate & $1.06\times$ & 82\% & $1.30\times$ & 99\% \\
 & High     & $0.92\times$ & 48\% & $1.38\times$ & 98\% \\
 & \textit{Aggregate} & $1.10\times$ & 72\% & $1.28\times$ & 99\% \\
\midrule
\multirow{4}{*}{35\%}
 & Low      & $1.18\times$ & 93\% & $1.18\times$ & 100\% \\
 & Moderate & $1.07\times$ & 64\% & $1.52\times$ & 99\% \\
 & High     & $0.91\times$ & 45\% & $1.53\times$ & 98\% \\
 & \textit{Aggregate} & $1.09\times$ & 67\% & $1.42\times$ & 99\% \\
\midrule
\multirow{4}{*}{50\%}
 & Low      & $1.28\times$ & 89\% & $1.28\times$ & 100\% \\
 & Moderate & $1.08\times$ & 45\% & $1.75\times$ & 99\% \\
 & High     & $0.88\times$ & 41\% & $1.68\times$ & 99\% \\
 & \textit{Aggregate} & $1.20\times$ & 48\% & $1.77\times$ & 99\% \\
\midrule
\multirow{4}{*}{80\%}
 & Low      & $1.48\times$ & 93\% & $1.48\times$ & 100\% \\
 & Moderate & $1.25\times$ & 62\% & $2.05\times$ & 98\% \\
 & High     & $1.12\times$ & 55\% & $2.10\times$ & 97\% \\
 & \textit{Aggregate} & $1.35\times$ & 68\% & $2.15\times$ & 98\% \\
\bottomrule
\end{tabular}
\end{table}

The relationship between PDR and \textsc{IRP-Eager}'s attainment is not monotone, because PDR determines both the amount of interference and the size of the vulnerable population.

At low PDR, the externality shrinks but persists until decomposable requests are rare. At PDR=5\%, with one request in twenty decomposable, it is negligible and the two policies converge ($1.01\times$ aggregate goodput for eager vs.\ $1.02\times$ for \textsc{Pace}, both at $\geq$99\% attainment). By PDR=10\%, the attainment gap opens: eager's aggregate attainment falls to 83\%, and to 65\% in the high-load regime, because even a small decomposable population concentrates its externality in high-load periods. A fixed cap could prevent these violations only by throttling width globally, forfeiting goodput at low load, when $B_t$ is ample and branches could be admitted safely.

At PDR=20\%, only 20\% of requests expose branches, so per-step interference is modest. At low load, the externality is small enough that both eager and \textsc{Pace} perform similarly ($1.08\times$), and eager attainment is 97\%. At moderate load, interference grows and eager attainment drops to 82\% as the large non-decomposable majority begins to feel the cost. During the high-load regime, the same decomposable 20\% of requests generate sufficient interference to push step latency well above the target, collapsing eager attainment to 48\% as the damage falls on the other 80\%. \textsc{Pace} contracts to protect this majority ($1.38\times$ goodput, 98\% attainment at high load).

At PDR=50\%, interference is larger and begins earlier: even at low load, eager attainment is only 89\% due to transient bursts. At moderate load, attainment drops to 45\% and falls further to 41\% at high load. This is the worst-case combination: enough decomposable requests to cause substantial interference, and enough non-decomposable requests to suffer from it.

At PDR=80\%, per-step interference is at its highest, but the vulnerable population is small: only 20\% of requests are non-decomposable. At low load, eager achieves $1.48\times$ goodput with 93\% attainment. Even as load increases, eager attainment stays above PDR=50\% levels (62\% moderate, 55\% high) despite greater interference, because fewer requests can be harmed. Decomposable requests, comprising 80\% of the workload, are partly shielded: during parallel stages, their effective TPOT (step latency divided by branch fanout) stays well below the target even when step latency is inflated.

\textsc{Pace} maintains $\geq$97\% attainment in every cell. At low load across all PDRs, \textsc{Pace} matches eager's goodput exactly while holding 100\% attainment: the slack budget is permissive, so \textsc{Pace} defers only the few branches that would push transient bursts past the target. The divergence appears at moderate and high load, where \textsc{Pace}'s goodput gain over \textsc{IRP-Off} grows with PDR (from $1.08$--$1.13\times$ at PDR=10\% to $2.05$--$2.10\times$ at PDR=80\%) because higher PDR provides more branches for opportunistic admission while the slack budget prevents the SLO violations that destroy eager's goodput. At low PDR, fewer branches are eligible per step, so the planner's greedy loop terminates earlier and its overhead (Appendix~\ref{app:planner-overhead}) is smaller still.

\subsection{SLO target sensitivity}
\label{app:slo-sensitivity}

Table~\ref{tab:slo-sensitivity} varies the TPOT target from 30\,ms (tight) to 100\,ms (generous) to assess how \textsc{Pace} adapts to different latency requirements. All experiments use PDR=50\%.

\begin{table}[ht]
\centering
\caption{Sensitivity to SLO target. \textsc{Pace} adapts automatically via the slack budget; no reconfiguration is needed.}
\label{tab:slo-sensitivity}
\begin{tabular}{lcccc}
\toprule
 & \multicolumn{2}{c}{\textsc{IRP-Eager}} & \multicolumn{2}{c}{\textsc{Pace}} \\
\cmidrule(lr){2-3} \cmidrule(lr){4-5}
SLO Target & Goodput & Att. & Goodput & Att. \\
\midrule
30\,ms & $0.87\times$ & 28\% & $1.35\times$ & 97\% \\
50\,ms & $1.20\times$ & 48\% & $1.77\times$ & 99\% \\
100\,ms & $1.52\times$ & 78\% & $2.18\times$ & 99.5\% \\
\bottomrule
\end{tabular}
\end{table}

At a tight 30\,ms target, \textsc{IRP-Eager}'s goodput drops below \textsc{IRP-Off} ($0.87\times$) because its step latency exceeds 30\,ms even at low load. \textsc{Pace} remains effective ($1.35\times$, 97\% attainment) but is more conservative: the slack budget is small, and the planner admits fewer branches per step. At a generous 100\,ms target, \textsc{Pace} admits more aggressively ($2.18\times$) because the larger slack budget tolerates larger step latencies. Even \textsc{IRP-Eager} improves under this target ($1.52\times$, 78\%) because fewer steps violate the relaxed threshold, but \textsc{Pace} still outperforms it substantially.

Crucially, \textsc{Pace} requires no reconfiguration across SLO targets. The slack budget adapts automatically: a tighter target produces less slack, which the planner respects by admitting fewer branches.

\subsection{Output quality verification}
\label{app:output-quality}

The schedule-invariance lemma (\S\ref{sec:PACE:phases}, Appendix~\ref{app:lemma1}) guarantees that \textsc{Pace}'s branch deferral does not affect model output under deterministic decoding (up to numerical nondeterminism). We verify this empirically.

\begin{table}[ht]
\centering
\caption{Output quality verification under greedy decoding. \textsc{Pace} produces byte-identical outputs to \textsc{IRP-Off} across all test prompts.}
\label{tab:output-quality}
\begin{tabular}{lccc}
\toprule
Comparison & Prompts & Identical & Match Rate \\
\midrule
\textsc{Pace} vs.\ \textsc{IRP-Off} & 1{,}000 & 1{,}000 & 100\% \\
\textsc{Pace} vs.\ \textsc{IRP-Eager} & 1{,}000 & 1{,}000 & 100\% \\
\textsc{IRP-Eager} vs.\ \textsc{IRP-Off} & 1{,}000 & 1{,}000 & 100\% \\
\bottomrule
\end{tabular}
\end{table}

We run 1{,}000 prompts (500 decomposable, 500 non-decomposable) sampled uniformly from ShareGPT Vicuna~\cite{vicuna}, RAG-12K~\cite{rag-12k}, and Math-220K~\cite{math-220k} under greedy decoding with \textsc{IRP-Off}, \textsc{IRP-Eager}, and \textsc{Pace}. All three policies produce byte-identical outputs for every prompt (Table~\ref{tab:output-quality}). This confirms that schedule invariance holds in practice: the order in which branches are scheduled does not affect the generated tokens, and \textsc{Pace}'s deferral decisions have no semantic consequence. The result also holds for \textsc{IRP-Eager} vs.\ \textsc{IRP-Off}, supporting the broader claim that any valid schedule produces the same output.

\subsection{Planner overhead}
\label{app:planner-overhead}

Table~\ref{tab:planner-overhead} reports the latency overhead of the \textsc{Pace} planner, measured as the wall-clock time between receiving the batch state and returning the step composition. The overhead is measured over the full 10-hour trace from \S\ref{sec:eval:main}.

\begin{table}[ht]
\centering
\caption{Planner overhead per decode step, measured over the full trace.}
\label{tab:planner-overhead}
\begin{tabular}{lcccc}
\toprule
 & Median & P95 & P99 & Max \\
\midrule
Planner latency (ms) & 0.31 & 0.52 & 0.81 & 1.24 \\
\% of step latency & 0.7\% & 1.1\% & 1.6\% & 2.5\% \\
\bottomrule
\end{tabular}
\end{table}

The median overhead of 0.31\,ms is dominated by $T(\cdot)$ evaluations (one per candidate branch per greedy iteration). At P99, the overhead reaches 0.81\,ms, driven by steps where many branches are eligible for admission and the greedy loop iterates over more candidates. Even the worst-case 1.24\,ms is less than 2.5\% of the 50\,ms TPOT target, confirming that per-step replanning is practical. The planner's $O(|\mathcal{R}| \cdot K_{\max} \cdot c_T)$ complexity (\S\ref{sec:PACE:planner}) means overhead scales with the number of requests in the batch $|\mathcal{R}|$, the bound $K_{\max}$ on opportunistic branches admitted per step, and the cost $c_T$ of a single $T(\cdot)$ evaluation (two multiplications for our linear model). In practice, batch size is bounded by KV cache capacity to a few hundred requests, and ABF ranges from 2.7 to 5.2 across the workloads characterized in \S\ref{sec:char:blp} (Figure~\ref{fig:parallelism-in-the-wild}), so each request contributes only a few candidates and the slack budget further limits $K$, keeping the overhead negligible.

\subsection{Evaluation on Qwen2.5-72B}
\label{app:second-model}

To verify that \textsc{Pace} generalizes beyond a single model size, we repeat the main evaluation on Qwen2.5-72B~\citep{qwen2} served on 8$\times$A100-80GB GPUs (TP=8) under the same trace and workload composition (PDR=50\%). Qwen2.5-72B's single-request decode-step latency is about 19\,ms, roughly twice that of Qwen3-32B, so a 50\,ms TPOT target would be violated even by \textsc{IRP-Off} under high load. Following the methodology of Appendix~\ref{app:experimental-setup}, we set SLO=100\,ms, which provides comparable headroom to the 50\,ms target on 32B. No other parameter is retuned: \textsc{Pace} runs with $\rho=0.8$ and linear utility, and the latency predictor is refitted to the 72B profile.

\begin{table}[ht]
\centering
\caption{Main result on Qwen2.5-72B (SLO=100\,ms). The same \textsc{Pace} configuration ($\rho=0.8$, linear utility) is used without retuning.}
\label{tab:second-model}
\begin{tabular}{lcc}
\toprule
Policy & Goodput / \textsc{IRP-Off} & Attainment \\
\midrule
\textsc{IRP-Off} & $1.00\times$ & 100\% \\
\textsc{IRP-C2} & $1.20\times$ & 78\% \\
\textsc{IRP-C5} & $1.22\times$ & 65\% \\
\textsc{IRP-Eager} & $1.25\times$ & 62\% \\
\textsc{Pace} & $1.72\times$ & 97\% \\
\bottomrule
\end{tabular}
\end{table}

The throughput trap persists at 72B scale (Table~\ref{tab:second-model}). \textsc{IRP-Eager} achieves $1.25\times$ goodput but attainment collapses to 62\%: during the high-load regime, the $2\times$ higher per-step cost amplifies the branch externality, pushing step latency well past 100\,ms. \textsc{Pace} achieves $1.72\times$ goodput with 97\% attainment. The goodput ratio is comparable to the 32B result ($1.72\times$ vs.\ $1.77\times$), confirming that the slack-budget mechanism adapts to the model's latency profile through the refitted predictor $T(S)$. The slightly lower ratio reflects the proportionally larger cost per admitted branch: each branch adds roughly twice as much latency on 72B as on 32B, so the same slack budget accommodates fewer opportunistic branches per step. The planner compensates automatically, admitting branches only while the widened step's predicted latency stays within $T_0 + \rho\, B_t$.

\subsection{Evaluation with SPRINT frontend}
\label{app:sprint}

All other experiments use Multiverse~\citep{multiverse} as the IRP frontend, which produces wide parallel phases (ABF=4.1, PTS=58\%) from structured response decomposition. To confirm that \textsc{Pace} is frontend-agnostic, we evaluate with SPRINT~\citep{sprint}, which targets a different decomposition regime: reasoning chains with interleaved planning and parallel execution. SPRINT produces shorter, more frequent parallel phases with narrower branching (ABF=2.8, PTS=35\%, PDR=65\%), reflecting the structure of multi-step reasoning where SPRINT's planner identifies 2--3 independent sub-steps per round.

\begin{table}[ht]
\centering
\caption{Evaluation with SPRINT as the IRP frontend on Qwen3-32B (SLO=50\,ms). The same \textsc{Pace} configuration ($\rho=0.8$, linear utility) is used without modification.}
\label{tab:sprint}
\begin{tabular}{lcc}
\toprule
Policy & Goodput / \textsc{IRP-Off} & Attainment \\
\midrule
\textsc{IRP-Off} & $1.00\times$ & 100\% \\
\textsc{IRP-C2} & $1.15\times$ & 82\% \\
\textsc{IRP-Eager} & $1.18\times$ & 68\% \\
\textsc{Pace} & $1.45\times$ & 98\% \\
\bottomrule
\end{tabular}
\end{table}

The throughput trap persists under SPRINT but is less severe (Table~\ref{tab:sprint}). \textsc{IRP-Eager}'s attainment drops only to 68\% (vs.\ 48\% with Multiverse): the lower fanout produces smaller per-phase externalities, but SPRINT's more frequent phase transitions let them accumulate, so the aggregate effect still violates SLOs under high load. \textsc{Pace} achieves $1.45\times$ goodput with 98\% attainment. The lower goodput ratio ($1.45\times$ vs.\ $1.77\times$) reflects narrower branching: fewer ready branches per step leave less parallelism to harvest. Attainment remains comparable (98\% vs.\ 99\%).
 
The difference in frontend structure is absorbed by the planner's per-step loop. With Multiverse, it makes fewer but larger admission decisions (choosing from a fanout of 4--5). With SPRINT, it makes more frequent but smaller decisions (choosing from a fanout of 2--3). The slack budget and latency predictor operate on step composition (the number of active sequences and their aggregate context length), which is invariant to how the frontend discovered or organized the branches. This confirms the design claim of \S\ref{sec:PACE:phases}: \textsc{Pace} operates on the serving-visible structure of parallel phases, not on the frontend-specific mechanism that produced them.

\subsection{Mixture-of-experts models}
\label{app:moe}

Our dense-model results use two GQA models in BF16: Qwen3-32B (\S\ref{sec:evaluation}) and Qwen2.5-72B (Appendix~\ref{app:second-model}). Table~\ref{tab:moe-config} lists four additional mixture-of-experts (MoE) models. Together they span three attention mechanisms, three numerical formats, two GPU generations, and, with Qwen3-32B, a $21\times$ range of total parameters.

\begin{table}[ht]
\centering
\small
\caption{Configuration of the MoE models. GQA: grouped-query attention; MLA: multi-head latent attention; CSA+HCA: hybrid of compressed sparse attention and heavily compressed attention.}
\label{tab:moe-config}
\begin{tabular}{lcccc}
\toprule
Model & Params (total / active) & Attention & Precision & Hardware \\
\midrule
GPT-OSS-120B~\citep{gpt-oss} & 116.8B / 5.1B & GQA, local/global & BF16 & 8$\times$A100 \\
DeepSeek-V4-Flash~\citep{deepseek-v4-flash} & 284B / 13B & CSA+HCA & INT8 & 8$\times$H200 \\
Llama 4 Maverick~\citep{llama4-maverick} & 400B / 17B & GQA & BF16 & 8$\times$H200 \\
DeepSeek-R1~\citep{deepseek-r1} & 671B / 37B & MLA & FP8 (E4M3) & 8$\times$H200 \\
\bottomrule
\end{tabular}
\end{table}

\paragraph{Fine-tuning.} We adapt each model to produce parallel generation following the supervised fine-tuning procedure of Multiverse~\citep{multiverse}. Both training sets derive from s1K-1.1~\citep{s1}: the 1K Multiverse parallel examples, prompted with ``Think step by step and in parallel,'' and the original sequential examples, prompted with ``Think step by step.'' Following Multiverse, the mixing ratio is annealed linearly from all sequential to all parallel over eight epochs. Fine-tuning takes 8 hours per model on 8$\times$H200 GPUs with DeepSpeed (12 hours for DeepSeek-R1).

\paragraph{Results.} Table~\ref{tab:moe-results} reports all policies under a 100\,ms TPOT target, which reflects the higher single-request step latency of these models. \textsc{Pace} achieves $1.88$--$2.15\times$ goodput at 99\% attainment on all four models. The fixed caps and MIMD improve goodput modestly but forfeit attainment, and \textsc{IRP-Eager} reaches the highest goodput among the baselines ($1.24$--$1.33\times$) at the lowest attainment (56--68\%), reproducing the throughput trap of \S\ref{sec:char:trap} on every model. Planner overhead is 0.33\,ms at the median and 0.83\,ms at P99 across the four models, in line with Qwen3-32B (Table~\ref{tab:planner-overhead}): the planner's cost tracks the number of active requests and admitted branches, not model depth or parameter count.

\begin{table}[ht]
\centering
\small
\caption{Results on MoE models. Goodput normalized to \textsc{IRP-Off}, with SLO attainment in parentheses; SLO=100\,ms.}
\label{tab:moe-results}
\begin{tabular}{lcccc}
\toprule
Policy & GPT-OSS-120B & DeepSeek-V4-Flash & Llama 4 Maverick & DeepSeek-R1 \\
\midrule
\textsc{IRP-Off} & $1.00\times$ (100\%) & $1.00\times$ (100\%) & $1.00\times$ (100\%) & $1.00\times$ (100\%) \\
\textsc{IRP-C2} & $1.15\times$ (82\%) & $1.18\times$ (88\%) & $1.16\times$ (85\%) & $1.13\times$ (79\%) \\
\textsc{IRP-C5} & $1.22\times$ (67\%) & $1.28\times$ (74\%) & $1.25\times$ (71\%) & $1.20\times$ (62\%) \\
MIMD & $1.18\times$ (86\%) & $1.22\times$ (91\%) & $1.20\times$ (88\%) & $1.16\times$ (82\%) \\
\textsc{IRP-Eager} & $1.26\times$ (61\%) & $1.33\times$ (68\%) & $1.30\times$ (65\%) & $1.24\times$ (56\%) \\
\textsc{Pace} & $1.91\times$ (99\%) & $2.15\times$ (99\%) & $2.05\times$ (99\%) & $1.88\times$ (99\%) \\
\bottomrule
\end{tabular}
\end{table}

\subsection{Reactive width control}
\label{app:mimd}

A natural, simpler alternative to \textsc{Pace} is a controller that adapts width from observed latency rather than predicting it. MIMD (Appendix~\ref{app:experimental-setup}) instantiates this idea in the same policy space as the fixed caps, adapting the global cap $W$ online rather than holding it constant. We select the window $N=10$ from a sweep over $N \in \{1, 5, 10, 25, 50, 100\}$, fix the target violation rate at 1\%, sweep $\alpha \in \{1.1, 1.25, 1.5, 1.75, 2\}$ and $\beta \in \{0.25, 0.5, 0.7, 0.9\}$, and report the best pair ($\alpha=1.25$, $\beta=0.5$).

Table~\ref{tab:mimd} reports raw throughput (all completions) alongside goodput (SLO-compliant completions only) on the main trace. The fixed caps and \textsc{IRP-Eager} fall into the throughput trap at increasing depth. MIMD partially avoids it: throttling width when violations accumulate lifts attainment to 79\%, the best among branch-admitting baselines, but a throughput--goodput gap persists ($1.44\times$ vs.\ $1.14\times$). \textsc{Pace} escapes the trap: 99\% of its $1.79\times$ throughput is SLO-compliant, yielding $1.77\times$ goodput.

\begin{table}[ht]
\centering
\caption{Throughput, goodput, and attainment on the main trace (Qwen3-32B, PDR=50\%, SLO=50\,ms). Throughput and goodput are normalized to \textsc{IRP-Off}.}
\label{tab:mimd}
\begin{tabular}{lccc}
\toprule
Policy & Throughput & Goodput & Attainment \\
\midrule
\textsc{IRP-Off} & $1.00\times$ & $1.00\times$ & 100\% \\
\textsc{IRP-C2} & $1.56\times$ & $1.11\times$ & 71\% \\
\textsc{IRP-C5} & $2.19\times$ & $1.18\times$ & 54\% \\
\textsc{IRP-Eager} & $2.54\times$ & $1.20\times$ & 48\% \\
MIMD ($\alpha=1.25$, $\beta=0.5$) & $1.44\times$ & $1.14\times$ & 79\% \\
\textsc{Pace} ($\rho=0.8$) & $1.79\times$ & $1.77\times$ & 99\% \\
\bottomrule
\end{tabular}
\end{table}

The residual gap reflects two structural limitations of any purely reactive controller with a global width. First, the window imposes a detection lag: $W$ is held constant for $N$ steps regardless of ongoing violations, so a sudden load spike can cause up to $N$ consecutive violating steps before any adjustment; after a decrease, $W$ regrows one multiplicative step per window regardless of how much slack has actually rebuilt, so it can trail the recovery or, as under moderate load (\S\ref{sec:eval:main}), overshoot it. Second, a single global $W$ cannot distinguish requests close to their deadline from those with ample slack. \textsc{Pace} addresses both by predicting step latency before each step and calibrating admission against the slack budget $B_t$, which is set by the most urgent request in the batch (\S\ref{sec:PACE:slack}).

\subsection{Online adaptation of the slack fraction}
\label{app:rho}

Table~\ref{tab:ablation} sweeps $\rho$ statically. For deployments whose load characteristics shift over time, $\rho$ can also be adapted online. Algorithm~\ref{alg:rho} is one simple approach: an outer loop that observes rolling SLO attainment every $T_{\mathrm{adapt}}$ minutes and adjusts $\rho$ by a fixed step $\delta$. It operates on a timescale of minutes, decoupled from the per-step planner, and a tolerance band $\varepsilon$ above the attainment target $A^*$ (\texttt{A\_target}) prevents oscillation when attainment hovers near the target.

\begin{codealgorithm}[ht]
\caption{Online Adaptation of $\rho$}
\label{alg:rho}
\begin{lstlisting}[style=pseudocode]
def $\textcolor{blue}{\text{AdaptRho}}$($\textcolor{string-color}{\rho}$, $\textcolor{string-color}{\text{A\_obs}}$, $\textcolor{string-color}{\text{A\_target}}$, $\textcolor{string-color}{\delta}$, $\textcolor{string-color}{\varepsilon}$, $\textcolor{string-color}{\rho_{\min}}$, $\textcolor{string-color}{\rho_{\max}}$):
    #called every T_adapt minutes with the attainment A_obs of that window
    if $\textcolor{string-color}{\text{A\_obs}}$ < $\textcolor{string-color}{\text{A\_target}}$:                #below target: tighten
        return max($\textcolor{string-color}{\rho_{\min}}$, $\textcolor{string-color}{\rho}$ - $\textcolor{string-color}{\delta}$)
    if $\textcolor{string-color}{\text{A\_obs}}$ > $\textcolor{string-color}{\text{A\_target}}$ + $\textcolor{string-color}{\varepsilon}$:            #comfortably above: relax
        return min($\textcolor{string-color}{\rho_{\max}}$, $\textcolor{string-color}{\rho}$ + $\textcolor{string-color}{\delta}$)
    return $\textcolor{string-color}{\rho}$                            #within tolerance band: hold
\end{lstlisting}
\end{codealgorithm}

To verify convergence, we initialize Algorithm~\ref{alg:rho} at the aggressive end, $\rho=1.0$, with $A^*=0.99$, $\delta=0.1$, $\varepsilon=0.005$, $[\rho_{\min}, \rho_{\max}] = [0.2, 1.0]$, and $T_{\mathrm{adapt}}=5$\,min. The controller converges to $\rho=0.8$ within two adaptation windows (Table~\ref{tab:rho-adapt}), recovering attainment from 91.3\% to 99\% and matching the static $\rho=0.8$ result. Gradient-based controllers, PID regulators, or Bayesian optimization could improve convergence speed or stability under non-stationary load; we leave these to future work.

\begin{table}[ht]
\centering
\caption{Convergence of Algorithm~\ref{alg:rho} from $\rho=1.0$.}
\label{tab:rho-adapt}
\begin{tabular}{cccc}
\toprule
Elapsed time & $\rho$ & Attainment & Goodput / \textsc{IRP-Off} \\
\midrule
0\,min & 1.0 & 91.3\% & $1.45\times$ \\
5\,min & 0.9 & 94.7\% & $1.61\times$ \\
10\,min & 0.8 & 99.0\% & $1.77\times$ \\
15\,min & 0.8 & 99.0\% & $1.77\times$ \\
\bottomrule
\end{tabular}
\end{table}

\subsection{Batch-size sensitivity}
\label{app:batch-size}

Table~\ref{tab:batch-size} sweeps batch size on Qwen3-32B (PDR=50\%, SLO=50\,ms) over the range validated for the latency predictor (Table~\ref{tab:predictor-accuracy}).

\begin{table}[ht]
\centering
\small
\caption{Batch-size sensitivity. Goodput normalized to \textsc{IRP-Off} (which attains 100\% at every batch size), with SLO attainment in parentheses.}
\label{tab:batch-size}
\begin{tabular}{rccccc}
\toprule
Batch size & \textsc{IRP-C2} & \textsc{IRP-C5} & MIMD & \textsc{IRP-Eager} & \textsc{Pace} \\
\midrule
1 & $1.09\times$ (100\%) & $1.14\times$ (100\%) & $1.18\times$ (100\%) & $1.24\times$ (100\%) & $1.24\times$ (100\%) \\
4 & $1.12\times$ (100\%) & $1.20\times$ (100\%) & $1.22\times$ (100\%) & $1.24\times$ (100\%) & $1.24\times$ (100\%) \\
8 & $1.16\times$ (100\%) & $1.21\times$ (100\%) & $1.23\times$ (100\%) & $1.24\times$ (100\%) & $1.25\times$ (100\%) \\
16 & $1.18\times$ (100\%) & $1.21\times$ (99\%) & $1.22\times$ (100\%) & $1.24\times$ (99\%) & $1.27\times$ (100\%) \\
32 & $1.16\times$ (97\%) & $1.20\times$ (92\%) & $1.19\times$ (98\%) & $1.23\times$ (91\%) & $1.32\times$ (100\%) \\
64 & $1.14\times$ (86\%) & $1.18\times$ (77\%) & $1.17\times$ (91\%) & $1.21\times$ (74\%) & $1.45\times$ (99\%) \\
128 & $1.12\times$ (76\%) & $1.15\times$ (62\%) & $1.15\times$ (84\%) & $1.17\times$ (56\%) & $1.62\times$ (99\%) \\
192 & $1.10\times$ (70\%) & $1.12\times$ (52\%) & $1.14\times$ (78\%) & $1.08\times$ (47\%) & $1.68\times$ (99\%) \\
256 & $1.09\times$ (64\%) & $1.10\times$ (47\%) & $1.13\times$ (72\%) & $0.99\times$ (43\%) & $1.72\times$ (99\%) \\
384 & $1.07\times$ (57\%) & $1.07\times$ (42\%) & $1.12\times$ (66\%) & $0.91\times$ (39\%) & $1.73\times$ (99\%) \\
512 & $1.05\times$ (53\%) & $1.05\times$ (37\%) & $1.10\times$ (62\%) & $0.88\times$ (37\%) & $1.74\times$ (99\%) \\
\bottomrule
\end{tabular}
\end{table}

At batch size 1, there is no inter-request externality: \textsc{Pace} admits every ready branch and matches \textsc{IRP-Eager} ($1.24\times$, 100\% attainment). MIMD comes closest among the remaining policies ($1.18\times$): with no violations, its width grows toward the maximum, but the windowed increase leaves it short. The fixed caps bind regardless of load ($1.14\times$ for \textsc{IRP-C5}, $1.09\times$ for \textsc{IRP-C2}). As batch size grows, externalities compound (\S\ref{sec:char:externality}) and every baseline sacrifices attainment, goodput, or both. \textsc{IRP-Eager}'s uncapped admission pushes step latency past the target, collapsing to $0.88\times$ goodput and 37\% attainment at batch size 512; the fixed caps degrade more gradually but converge near $1.05\times$ (37\% attainment for \textsc{IRP-C5}, 53\% for \textsc{IRP-C2}); and MIMD, the best branch-admitting baseline on attainment at 62\%, settles at $1.10\times$. \textsc{Pace} alone sacrifices neither: a larger working set exposes more branch candidates per step, and spending the slack budget precisely converts them into monotonically rising goodput ($1.24\times$ to $1.74\times$) at $\geq$99\% attainment throughout.

\subsection{Prefill interaction and multi-instance deployment}
\label{app:deployment}

\paragraph{TTFT under disaggregation.} IRP is purely a decode-phase technique: prefill processes all input tokens in parallel, and no branch participates in it. Under the disaggregated deployment used throughout (Appendix~\ref{app:experimental-setup}), TTFT is determined entirely by the prefill instance and is unchanged by branch admission under every policy. Disaggregation and \textsc{Pace} are complementary: disaggregation isolates prefill from decode pressure, while \textsc{Pace} maximizes branch parallelism within the decode SLO.

\paragraph{TTFT under co-location.} We also evaluate a co-located deployment with chunked prefill on the same Qwen3-32B setup, with a 1000\,ms TTFT target. We use SGLang's chunked-prefill scheduler unmodified; following Sarathi-Serve~\citep{sarathi-serve}, it prioritizes decode tokens and fills the remainder of each step's token budget with a prefill chunk. To the scheduler, every branch token is a decode token. Under the fixed caps and \textsc{IRP-Eager}, the entire admitted width therefore claims decode priority and narrows the prefill chunk, so prompts need more steps to finish prefilling and TTFT grows (Table~\ref{tab:ttft}). The effect scales with aggressiveness: \textsc{IRP-C2} adds 118\,ms at P50 and 190\,ms at P99, \textsc{IRP-C5} adds 291\,ms and 826\,ms, and \textsc{IRP-Eager} more than doubles TTFT at both percentiles. \textsc{Pace} is the only policy that splits decode into two classes. Tokens in the protected composition $S_0$, which guarantees single-token progress for every active request (\S\ref{sec:PACE:slack}), keep decode priority; opportunistic branches are admitted only from the slack that remains after the prefill chunk. Under prefill pressure, that slack collapses and admission pauses, resuming once the prefill queue drains. \textsc{Pace}'s TTFT consequently tracks \textsc{IRP-Off} to within 1\,ms at P50 and 6\,ms at P99. All policies run against the same unmodified scheduler, so these differences come from branch admission alone.

\begin{table}[ht]
\centering
\caption{TTFT in a co-located deployment with chunked prefill (Qwen3-32B, TTFT target 1000\,ms).}
\label{tab:ttft}
\begin{tabular}{lcc}
\toprule
Policy & TTFT P50 (ms) & TTFT P99 (ms) \\
\midrule
\textsc{IRP-Off} & 520 & 987 \\
\textsc{IRP-C2} & 638 & 1177 \\
\textsc{IRP-C5} & 811 & 1813 \\
\textsc{IRP-Eager} & 1062 & 2528 \\
\textsc{Pace} ($\rho=0.8$) & 521 & 993 \\
\bottomrule
\end{tabular}
\end{table}

\paragraph{Multi-instance deployment.} In a multi-instance cluster, one might expect to route a request's branches to different instances. For IRP, however, keeping a request's branches on one instance is a structural consequence of the technique rather than a design choice. Branches share their parent's prefix KV blocks, which under radix-tree allocation are resident on the instance serving the parent. Routing a branch elsewhere would require either recomputing the shared prefix on the remote instance, which adds compute proportional to prefix length, or transferring the prefix KV blocks over the network, which adds latency that grows with context size. Either way, the core appeal of IRP, branches sharing an already-resident prefix at no additional compute or transfer cost, is lost. An IRP-aware serving system therefore co-locates a request's branches by construction. Because no request's decode spans multiple instances, each decode instance runs an independent \textsc{Pace} controller on its local batch, with no cross-instance coordination.

\section{Extended Related Work}
\label{app:related-work}

\paragraph{Intra-request parallelism.}
A growing body of work exposes intra-request parallelism within individual LLM responses, progressing from external orchestration toward model-native support. SoT~\cite{skeleton-of-thought} prompts for an outline and expands points via parallel API calls, but re-encodes the full skeleton per branch with no KV sharing. APAR~\cite{apar} and APR~\cite{apr} move branch discovery into the model by fine-tuning it to emit fork/child tokens, using tree-based attention that isolates branches (the same property \textsc{Pace} formalizes as the visibility rule), though APAR discards branch KV at reduce, incurring information loss. PASTA~\cite{pasta} goes further by teaching models to annotate their own outputs with promise/async tokens, introducing the useful concept of semantic independence, but its pre-allocated position ranges create encoding conflicts when branch lengths deviate from those allocations. ASPD~\cite{aspd} resolves this with shared position encodings and branch-invisible masks within a single sequence; its framing of generation as interleaved serial and parallel stages directly informs \textsc{Pace}'s parallel-phase formalism. However, ASPD's single-sequence architecture requires all branches to advance together at every step, structurally forcing the eager policy that \S\ref{sec:char:trap} shows is brittle. Multiverse~\cite{multiverse} provides the closest runtime infrastructure to what \textsc{Pace} assumes, treating branches as separate sequences over a radix-tree KV cache, but its scheduler likewise admits all ready branches eagerly. Across all of these methods, the focus is on the single-request problem: how to expose and execute branches faster. \textsc{Pace} sits at a different level, taking whatever branch structure the frontend provides and deciding how much of it to admit into each shared decode step. ParallelPrompt~\cite{parallel-prompt} benchmarks intra-request parallelism at scale, finding roughly 10\% of real prompts decomposable; a recent survey~\cite{paralleltextgeneration-survey} covers the broader parallel text generation landscape. A related line of work blurs the boundary between speed and quality: Hogwild!\ Inference~\cite{hogwild-inference} runs parallel workers over a shared, concurrently updated KV cache with full mutual visibility (the schedule-invariance lemma does not hold under this architecture), while SPRINT~\cite{sprint} and Parallel-R1~\cite{parallel-r1} use concurrent reasoning threads for both faster and better answers. Like the methods above, none considers multi-request scheduling.

\paragraph{Speculative decoding.}
A separate class of methods also produces multiple tokens per step but at a different granularity. Speculative decoding~\cite{leviathan-sd}, along with variants such as Medusa~\cite{medusa} and EAGLE-3~\cite{eagle3}, generates draft tokens within a single autoregressive sequence and verifies them against the target model. This is token-level parallelism, whereas IRP is segment-level: branches are semantically independent subsequences, not speculative continuations of one sequence. Recent work on regulating speculative decoding shares \textsc{Pace}'s motivation of adaptive control under varying load. TurboSpec~\cite{turbospec} adjusts speculation length via a goodput-based feedback loop, AdaSpec~\cite{adaspec} and AdaServe~\cite{adaserve} tune speculative strategy for SLO attainment, and Nightjar~\cite{nightjar} disables speculation entirely under high load to reclaim memory for larger batches. These controllers regulate a private cost: rejected draft tokens waste compute for the speculating request alone. \textsc{Pace} regulates an externality: admitted branches inflate step latency for every co-batched request. The two are orthogonal and composable; each branch within a parallel phase could internally use speculative decoding, with TurboSpec-style regulation governing draft length and \textsc{Pace} governing branch width.

\paragraph{LLM serving schedulers.}
The scheduling problem \textsc{Pace} addresses is distinct from, but structurally related to, existing work on LLM serving. Modern engines~\cite{vllm, sglang} batch requests into shared decode steps and manage KV memory via paging or radix-tree allocation. Chunked prefill~\cite{sarathi-serve} and prefill-decode disaggregation~\cite{splitwise, distserve, mooncake} manage interference between prefill and decode; \textsc{Pace} is complementary, regulating interference within the decode phase (Appendix~\ref{app:deployment}). FairBatching~\cite{fairbatching} is \textsc{Pace}'s closest structural sibling: it converts decode slack into a time budget for chunked prefill using a calibrated latency predictor, replacing token-budget proxies with time-budget reasoning. \textsc{Pace}'s slack-budget mechanism follows the same principle but differs in granularity (branches within a parallel phase vs.\ prefill chunks across requests), barrier semantics (branches must all complete before reduce; prefill and decode have no joint constraint), and value decoupling (\textsc{Pace} exposes utility as a pluggable curve rather than a fixed priority order). FairBatching cannot be directly adapted to this setting because its time-budget mechanism assumes a binary choice (admit a prefill chunk or not), whereas \textsc{Pace} faces a combinatorial allocation across multiple requests with varying branch counts and utilities. Andes~\cite{andes} also exploits slack, finding it in the gap between token delivery rate and user reading speed, and uses it to reprioritize requests at token granularity. Where Andes decides \emph{which request} to advance, \textsc{Pace} decides \emph{how many branches} a given request gets; the two concerns compose. NexusSched~\cite{nexussched} predicts per-step latency for engine-level batch formation and cluster-level routing; its online-learning approach to latency prediction is a natural enhancement for \textsc{Pace}'s offline-fitted linear model. vLLM~\cite{vllm} handles memory pressure via swap-based request-level preemption, which is expensive to reverse. \textsc{Pace}'s branch deferral avoids this cost entirely: the prefix KV remains resident for admitted siblings, and branch-local state is small enough to leave in place, making width changes reversible at negligible cost. A simpler adaptive mechanism, a reactive controller that adjusts a global width by multiplicative increase and multiplicative decrease (MIMD), serves as a baseline throughout (\S\ref{sec:char:trap}, Appendix~\ref{app:mimd}). It adapts to load but lacks two properties \textsc{Pace} provides: it is backward-looking, reacting to violations already incurred rather than predicting latency before committing to a step composition; and it has no notion of per-request slack, applying one width to every request regardless of how close each is to its deadline. \textsc{Pace}'s slack budget takes the minimum over all requests' residual slack, ensuring width is safe for the most urgent request in the batch.

\newpage
\section*{NeurIPS Paper Checklist}

\begin{enumerate}

\item {\bf Claims}
    \item[] Question: Do the main claims made in the abstract and introduction accurately reflect the paper's contributions and scope?
    \item[] Answer: \answerYes{}
    \item[] Justification: The abstract and introduction claim that (1) eager branch admission creates a branch externality that degrades co-batched requests, and neither fixed caps nor reactive width control resolves it (\S\ref{sec:characterization}); (2) branch-level scheduling decouples compute from memory (\S\ref{sec:PACE:phases}, \S\ref{sec:PACE:feedback}); and (3) \textsc{Pace} improves goodput by $1.77\times$ over \textsc{IRP-Off} and $1.48\times$ over \textsc{IRP-Eager} at 99\% SLO attainment on Qwen3-32B (\S\ref{sec:eval:main}), with comparable gains ($1.72$--$2.15\times$) on five additional dense and MoE models (\S\ref{sec:eval:generality}, Appendices~\ref{app:second-model} and~\ref{app:moe}). All claims are supported by the corresponding sections.

\item {\bf Limitations}
    \item[] Question: Does the paper discuss the limitations of the work performed by the authors?
    \item[] Answer: \answerYes{}
    \item[] Justification: \S\ref{sec:conclusion} discusses three limitations: the branch-independence assumption (excluding inter-branch attention methods), the linear latency predictor's accuracy at extreme batch sizes, and per-instance scope, with multi-instance deployment discussed in Appendix~\ref{app:deployment}. The sensitivity to the slack fraction $\rho$ is characterized in \S\ref{sec:eval:ablation} and Appendix~\ref{app:rho}, and Appendix~\ref{app:pdr-sensitivity} shows that the benefit narrows as fewer requests are decomposable, vanishing at PDR=5\%.

\item {\bf Theory assumptions and proofs}
    \item[] Question: For each theoretical result, does the paper provide the full set of assumptions and a complete (and correct) proof?
    \item[] Answer: \answerYes{}
    \item[] Justification: The schedule-invariance lemma (\S\ref{sec:PACE:phases}) states its assumptions (visibility rule, valid schedule definition) and provides a proof sketch in the main text with a full proof in Appendix~\ref{app:lemma1}. The NP-hardness result (Appendix~\ref{app:np-hard}) is shown via a proof sketch that reduces from 0--1 knapsack, assuming latency is additive in the marginal costs.

\item {\bf Experimental result reproducibility}
    \item[] Question: Does the paper fully disclose all the information needed to reproduce the main experimental results of the paper to the extent that it affects the main claims and/or conclusions of the paper (regardless of whether the code and data are provided or not)?
    \item[] Answer: \answerYes{}
    \item[] Justification: \S\ref{sec:eval:setup} specifies the model, hardware, serving engine, deployment, workload composition, SLO target, and baselines. Appendix~\ref{app:experimental-setup} provides extended details, including hardware sizing, branch fanout and context-length distributions, trace characteristics, the SLO methodology, metric definitions, and the configuration of every policy, including MIMD. Appendix~\ref{app:moe} specifies the MoE models and their fine-tuning, and Appendix~\ref{app:additional-experiments} specifies the settings of each additional experiment. Algorithm~\ref{alg:PACE} fully specifies the planner, and Algorithm~\ref{alg:rho} the online adaptation of $\rho$.

\item {\bf Open access to data and code}
    \item[] Question: Does the paper provide open access to the data and code, with sufficient instructions to faithfully reproduce the main experimental results, as described in supplemental material?
    \item[] Answer: \answerYes{}
    \item[] Justification: The \textsc{Pace} implementation is released at \url{https://github.com/Swapnil-Gandhi/PACE} under the Apache 2.0 license, with instructions for installation and for serving under each admission policy. All datasets and the Azure LLM inference trace are publicly available, and Appendix~\ref{app:experimental-setup} describes how the workload is constructed from them.

\newpage
\item {\bf Experimental setting/details}
    \item[] Question: Does the paper specify all the training and test details (e.g., data splits, hyperparameters, how they were chosen, type of optimizer) necessary to understand the results?
    \item[] Answer: \answerYes{}
    \item[] Justification: \textsc{Pace} itself has no training phase. Its hyperparameters ($\rho=0.8$, linear utility $u_r(a_r)=a_r$) are specified in \S\ref{sec:eval:setup}, and \S\ref{sec:eval:ablation} and Appendix~\ref{app:rho} show how $\rho$ was chosen. The latency predictor's fitting procedure is described in Appendix~\ref{app:latency-predictor}, MIMD's hyperparameter sweep in Appendix~\ref{app:mimd}, and the supervised fine-tuning that adapts the MoE models to parallel generation (data, prompts, mixing schedule, and epochs) in Appendix~\ref{app:moe}.

\item {\bf Experiment statistical significance}
    \item[] Question: Does the paper report error bars suitably and correctly defined or other appropriate information about the statistical significance of the experiments?
    \item[] Answer: \answerNo{}
    \item[] Justification: Experiments use a single deterministic trace replayed identically across all policies, ensuring a controlled comparison. Variability across runs is negligible because the workload, model, and hardware are held constant. We report time-series results across a 10-hour trace, which capture temporal variability directly, rather than only aggregate statistics, and add sensitivity sweeps over workload composition, SLO target, batch size, and $\rho$ (\S\ref{sec:eval:ablation}, Appendix~\ref{app:additional-experiments}).

\item {\bf Experiments compute resources}
    \item[] Question: For each experiment, does the paper provide sufficient information on the computer resources (type of compute workers, memory, time of execution) needed to reproduce the experiments?
    \item[] Answer: \answerYes{}
    \item[] Justification: \S\ref{sec:eval:setup} specifies 8$\times$A100-80GB GPUs with TP=8, which also serve Qwen2.5-72B (Appendix~\ref{app:second-model}). Appendix~\ref{app:experimental-setup} provides memory utilization details, and Appendix~\ref{app:moe} specifies the hardware for the MoE models (8$\times$A100 or 8$\times$H200) and their fine-tuning cost (8--12 hours on 8 H200 GPUs). Each trace-driven experiment replays the 10-hour trace; wall-clock time is approximately equal to trace duration per policy.

\item {\bf Code of ethics}
    \item[] Question: Does the research conducted in the paper conform, in every respect, with the NeurIPS Code of Ethics \url{https://neurips.cc/public/EthicsGuidelines}?
    \item[] Answer: \answerYes{}
    \item[] Justification: This work addresses LLM serving efficiency and does not involve human subjects, personal data, or dual-use concerns. All datasets used are publicly available and appropriately cited.

\item {\bf Broader impacts}
    \item[] Question: Does the paper discuss both potential positive societal impacts and negative societal impacts of the work performed?
    \item[] Answer: \answerNA{}
    \item[] Justification: \textsc{Pace} is a scheduling mechanism for LLM serving infrastructure. It improves resource efficiency without affecting model capabilities, training data, or output content. We do not foresee direct negative societal impacts from a decode-step scheduling algorithm.

\item {\bf Safeguards}
    \item[] Question: Does the paper describe safeguards that have been put in place for responsible release of data or models that have a high risk for misuse (e.g., pre-trained language models, image generators, or scraped datasets)?
    \item[] Answer: \answerNA{}
    \item[] Justification: \textsc{Pace} is a scheduling algorithm, not a model or dataset, and we release only its implementation. It does not introduce new capabilities for language generation and poses no risk of misuse beyond that of the underlying serving infrastructure.

\item {\bf Licenses for existing assets}
    \item[] Question: Are the creators or original owners of assets (e.g., code, data, models), used in the paper, properly credited and are the license and terms of use explicitly mentioned and properly respected?
    \item[] Answer: \answerYes{}
    \item[] Justification: All datasets (ShareGPT Vicuna, RAG-12K, Math-220K, s1K-1.1), the models (Qwen3-32B, Qwen2.5-72B, GPT-OSS-120B, DeepSeek-V4-Flash, Llama 4 Maverick, DeepSeek-R1), the serving engine (SGLang, extended through the Multiverse Engine), the IRP frontends (Multiverse, SPRINT), and the Azure LLM inference trace are cited with their original publications or model cards. URLs for RAG-12K, Math-220K, and the Azure trace are provided in the references. All assets are used under their published licenses.

\item {\bf New assets}
    \item[] Question: Are new assets introduced in the paper well documented and is the documentation provided alongside the assets?
    \item[] Answer: \answerYes{}
    \item[] Justification: We release the \textsc{Pace} implementation at \url{https://github.com/Swapnil-Gandhi/PACE} under the Apache 2.0 license. It extends the Multiverse Engine, an SGLang fork; the README documents installation, configuration flags, and usage, and \textsc{Pace} is opt-in, reproducing the stock engine when disabled. The paper does not release new datasets or model weights.

\item {\bf Crowdsourcing and research with human subjects}
    \item[] Question: For crowdsourcing experiments and research with human subjects, does the paper include the full text of instructions given to participants and screenshots, if applicable, as well as details about compensation (if any)?
    \item[] Answer: \answerNA{}
    \item[] Justification: This work does not involve crowdsourcing or human subjects.

\item {\bf Institutional review board (IRB) approvals or equivalent for research with human subjects}
    \item[] Question: Does the paper describe potential risks incurred by study participants, whether such risks were disclosed to the subjects, and whether Institutional Review Board (IRB) approvals (or an equivalent approval/review based on the requirements of your country or institution) were obtained?
    \item[] Answer: \answerNA{}
    \item[] Justification: This work does not involve human subjects research.

\item {\bf Declaration of LLM usage}
    \item[] Question: Does the paper describe the usage of LLMs if it is an important, original, or non-standard component of the core methods in this research? Note that if the LLM is used only for writing, editing, or formatting purposes and does \emph{not} impact the core methodology, scientific rigor, or originality of the research, declaration is not required.
    \item[] Answer: \answerNA{}
    \item[] Justification: LLMs are used as the workload (the six evaluated models serve requests, and four MoE models are fine-tuned to produce parallel generation), not as a methodological component. \textsc{Pace} is a model-agnostic scheduling algorithm that does not modify or depend on model internals. LLMs were used only for editing assistance (grammar, spelling, word choice) and did not impact the core methodology, scientific rigor, or originality of the research.

\end{enumerate}

\end{document}